\documentclass[11pt,reqno]{amsart}
\usepackage{graphicx}

\usepackage{amsmath}
\usepackage{amsfonts}
\usepackage{amssymb}
\usepackage{xypic}
\usepackage{url}
\usepackage{enumerate}
\usepackage{lscape}
\usepackage[usenames,dvipsnames]{color}
\usepackage{multicol}
\usepackage{multirow}
\usepackage{hhline}
\usepackage{array}
\usepackage{color}
\usepackage{tikz}
\usetikzlibrary{arrows,positioning,matrix,decorations.pathmorphing,calc,fadings,decorations.pathreplacing,patterns,decorations.markings,shapes}
\usepackage[all,matrix,arrow,import,arc,color]{xy}
\usepackage{mathrsfs}
\usepackage{blkarray}

\usepackage[normalem]{ulem}

\usepackage{multirow}

\usepackage{xcolor}

\newtheorem{theorem}{Theorem}[section]
\newtheorem{corollary}[theorem]{Corollary}
\newtheorem{proposition}[theorem]{Proposition}
\newtheorem{lemma}[theorem]{Lemma}
\newtheorem{example}[theorem]{Example}

\newtheorem{definition}[theorem]{Definition}

\theoremstyle{remark}
\newtheorem{remark}[theorem]{Remark}

\theoremstyle{propsition}
\newtheorem{hypothesis}[theorem]{Assumption}
\newcommand\hypo{\textnormal{Assumption}}

\newcommand{\Z}{\mathbb{Z}}
\newcommand{\R}{\mathbb{R}}

\newcommand{\N}{\mathbb{N}}

\newcommand{\cM}{\mathcal{M}}
\newcommand{\cP}{\mathcal{P}}
\newcommand{\cC}{\mathcal{C}}
\newcommand{\cK}{\mathcal{K}}

\newcommand{\Sp}{\mathscr{S}}
\newcommand{\Edg}{\mathcal{R}}
\newcommand{\Ve}{\mathcal{V}}
\newcommand{\Ka}{\mathbf{k}}
\newcommand{\bx}{\mathbf{x}}

\newcommand{\Zp}{\mathscr{Z}}
\newcommand{\Wp}{\mathscr{W}}

\begin{document}

\title{Identifiability from a few species for a class of biochemical reaction networks}
\author[G. Jeronimo, M. P\'erez Mill\'an, P. Solern\'o]{Gabriela Jeronimo$^{1,2,3,\dag}$, Mercedes P\'erez Mill\'an$^{1,2,3,\dag}$, Pablo Solern\'o$^{1,2,\dag}$}
\address{$^1$ Universidad de Buenos Aires. Facultad de Ciencias Exactas y Naturales. Departamento de Matem\'atica. Ciudad Universitaria, Pab.\ I, C1428EGA Buenos Aires, Argentina.}
\address{$^2$ Universidad de Buenos Aires. Consejo Nacional de Investigaciones Cient\'ificas y T\'ecnicas. Instituto de Investigaciones Matem\'aticas ``Luis A. Santal\'o'' (IMAS). Facultad de Ciencias Exactas y Naturales, Ciudad Universitaria, Pab.\ I, C1428EGA Buenos Aires, Argentina.}
\address{$^3$ Universidad de Buenos Aires. Ciclo B\'asico Com\'un. Departamento de Ciencias Exactas. Buenos Aires, Argentina.}
\address{$^\dag$ These authors contributed equally to this work.}
\email{jeronimo@dm.uba.ar, mpmillan@dm.uba.ar, psolerno@dm.uba.ar}
\date{}

\maketitle

\begin{abstract}
Under mass-action kinetics, biochemical reaction networks give rise to polynomial autonomous dynamical  systems whose parameters are often difficult to estimate. We deal in this paper with the problem of identifying the kinetic parameters of a class of biochemical networks which are abundant, such as multisite phosphorylation systems and phosphorylation cascades (for example, MAPK cascades). For any system of this class we explicitly exhibit a single species for each connected component of the associated digraph such that the successive total derivatives of its concentration allow us to identify all the parameters occurring in the component. The number of derivatives needed is bounded essentially by the length of the corresponding connected component of the digraph. Moreover, in the particular case of the cascades, we show that the parameters can be identified from a bounded number of successive derivatives of the last product of the last layer.
This theoretical result induces also a heuristic interpolation-based identifiability procedure to recover the values of the rate constants from exact measurements.
\end{abstract}


\section{Introduction}
Parameter identifiability in a system of ordinary differential equations mainly addresses the question of deciding whether the system parameters can be uniquely determined from data (see for instance \cite{wp97},\cite[Chapter~10]{DiStefano14}).Since the pioneering paper \cite{BA70}, this problem  has been broadly studied for general systems under different perspectives, including Taylor series and generating series approaches, and differential algebra-based approaches. More details can be found in \cite{pohj78,oll90,LG94,Sed02,XM03,SAD03,BSAD07,MEDS09,chis11,Raue14,HOPY18a}. Also, a variety of software tools for identifiability has been developed that work for general classes of models (e.g. polynomial or rational), such as DAISY \cite{BSAD07}, COMBOS \cite{MKD14}, GenSSI \cite{LFC+17}, and SIAN \cite{HOPY18b}.

In this paper, we address the identifiability problem for a specific infinite class of models. Our aim is to obtain general statements about all the models in the class (see \cite{WE16,BME17} for prior results of this sort but for different classes of models).
More precisely, we consider a particular class of systems of equations arising from biochemical reaction networks under mass-action kinetics, which induces polynomial autonomous systems of differential equations. In this framework, in \cite{CP}, the authors describe necessary and sufficient conditions for the unique identifiability of the reaction rate constants (the parameters) of a chemical reaction network. Following their approach, we provide in this work sufficient conditions for uniquely identifying all the rate constants of a certain family of biochemical reaction networks from a reduced set of variables (see Definition~\ref{identvar}). Unlike other authors \cite{AKJ12}, we do not consider all the possible minimal sets of variables allowing parameter identifiability, but we only focus on certain biologically relevant sets.

The family of networks we deal with is abundant in the literature. One example is the multisite phosphorylation system which describes the phosphorylation of a protein in $L$ sites by a kinase($Y$)/phosphatase($\tilde{Y}$) pair in a sequential and distributive mechanism \cite{cyc-007}.
The substrate $S_i$ is the phosphoform obtained from the unphosphorylated substrate $S_0$ by attaching $i$ phosphate groups to it. Each phosphoform can accept (via an enzymatic reaction involving $Y$) or lose (via a reaction involving the phosphatase $\tilde{Y}$) at most one phosphate (the mechanism is ``distributive'') and there is a specific order to be followed for attaching and removing the phosphate groups (the phosphorylation is ``sequential'').

\begin{example} \label{multisite}
 The reactions in the $L$-site sequential phosphorylation/dephosphorylation network are represented by the following labeled digraph:
 
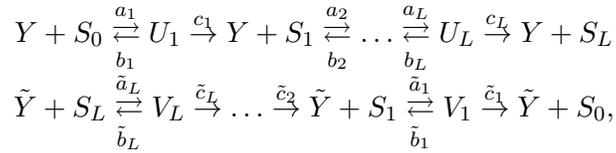
\[
 \begin{array}{c}
 Y+S_0 \overset{a_1}{\underset{b_1}{\rightleftarrows}} U_1 \overset{c_1}{\rightarrow}
 Y+S_1 \overset{a_2}{\underset{b_2}{\rightleftarrows}} \dots
 \overset{a_L}{\underset{b_L}{\rightleftarrows}}U_L \overset{c_L}{\rightarrow}
 Y+S_L\\
 \tilde{Y}+S_L \overset{\tilde{a}_L}{\underset{\tilde{b}_L}{\rightleftarrows}} V_L \overset{\tilde{c}_L}{\rightarrow} \dots \overset{\tilde{c}_2}{\rightarrow}
 \tilde{Y}+S_1 \overset{\tilde{a}_{1}}{\underset{\tilde{b}_1}{\rightleftarrows}}  V_1 \overset{\tilde{c}_1}{\rightarrow}
 \tilde{Y}+S_0,
 \end{array}\]
 where $U_1,\dots,U_L,V_1,\dots,V_L$ are intermediate enzyme-substrate species. The mass-action dynamical system for this network is (see identity \eqref{CRN} in Section \ref{subsec: chemical reaction}):
  \[\begin{array}{rcl}
 \dot{s}_0 & = &-a_1ys_0+b_1u_1+\tilde{c}_1v_1,\\
 \dot{s}_L & = &c_Lu_L-\tilde{a}_L\tilde{y}s_L+\tilde{b}_Lv_L,\\
 \dot{u_i} & = & a_iys_{i-1}-(b_i+c_i)u_i, \ 1\le i \le L,\\
 \dot{v}_i & = & \tilde{a}_i\tilde{y}s_i-(\tilde{b}_i+\tilde{c}_i)v_i, \ 1\le i \le L,\\
 \dot{y}  & = &\underset{i=1}{\overset{L}{\sum}}-a_iys_{i-1}+(b_i+c_i)u_i,\\
 \dot{\tilde{y}} & = &\underset{i=1}{\overset{L}{\sum}}-\tilde{a}_i\tilde{y}s_i+(\tilde{b}_i+\tilde{c}_i)v_i,\\
 \dot{s}_i& = &c_iu_i-a_{i+1}ys_i+b_{i+1}u_{i+1}+\tilde{c}_{i+1}v_{i+1}-\tilde{a}_i\tilde{y}s_i+\tilde{b}_iv_i, \; 1\leq i\leq L-1,
\end{array}\]
where lower case letters represent the time-varying concentration of the corresponding chemical species. Here, the derivative with respect to time is represented with a dot over the corresponding variable.
\end{example}
As a consequence of Theorem~\ref{thm:eachconncomp} proved below, all the constants in the first connected component can be identified, in the sense of Definition \ref{identvar}, from the successive total derivatives of $s_L$ up to order $\max\{2,2L-1\}$ and all the constants in the second connected component can be identified from the successive total derivatives of $s_0$ up to the same order. Moreover, as proved in Proposition~\ref{prop:2comp}, all the constants in the whole network can be identified from the successive total derivatives of $s_L$ up to order $\max\{2,2L-1\}$.

\medskip

Another example of major biological importance are phosphorylation cascades, such as the mitogen-activated protein kinase (MAPK) cascade \cite{CDVS16,sig-016,kholo00,sig-051}. This cascade plays an essential role in signal transduction by modulating gene transcription in response to changes in the cellular environment. MAPK cascades participate in a number of diseases including chronic inflammation and cancer \cite{davis,kyriakis,pearson,schaeffer,zarubin} as they control key cellular functions \cite{hornberg,pearson,widmann}.
We depict in the following example the $2$-layer signaling cascade.

\begin{example} \label{2layer} Consider the graph associated to the two-layer simple phosphorylation cascade where the simplified diagram and the corresponding reactions are, respectively:

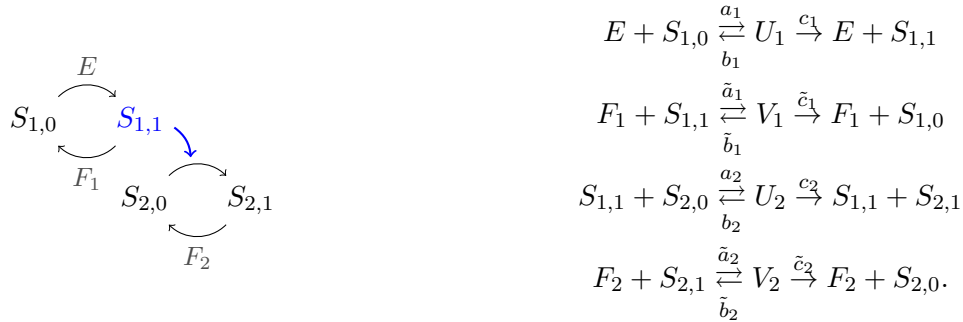
\begin{minipage}{0.5\textwidth}
 \[\begin{tikzpicture}[scale=0.7,node distance=0.5cm]
  \node[] at (1.1,-1.5) (dummy2) {};
  \node[left=of dummy2] (p0) {$S_{2,0}$};
  \node[right=of p0] (p1) {$S_{2,1}$}
    edge[->, bend left=45] node[below] {\textcolor{black!70}{\small{$F_2$}}} (p0)
    edge[<-, bend right=45] node[above] {} (p0);
  \node[] at (-1,0) (dummy) {};
  \node[left=of dummy] (s0) {$S_{1,0}$};
  \node[right=of s0] (s1) {\textcolor{blue}{$S_{1,1}$}}
    edge[->, bend left=45] node[below] {\textcolor{black!70}{\small{$F_1$}}} (s0)
    edge[<-, bend right=45] node[above] {\textcolor{black!70}{\small{$E$}}} (s0)
    edge[->,thick,color=blue, bend left=25] node[above] {} ($(p0.north)+(25pt,10pt)$);
 \end{tikzpicture}
\]
\end{minipage}
\begin{minipage}{0.5\textwidth}
 \[E+S_{1,0} \overset{a_1}{\underset{b_1}{\rightleftarrows}}
  U_1 \overset{c_1}{\rightarrow}
  E+S_{1,1}\]
\[F_1+S_{1,1} \overset{\tilde{a}_1}{\underset{\tilde{b}_1}{\rightleftarrows}}
  V_1 \overset{\tilde{c}_1}{\rightarrow}
  F_1+S_{1,0}\]
\[S_{1,1}+S_{2,0} \overset{a_2}{\underset{b_2}{\rightleftarrows}}
  U_2 \overset{c_2}{\rightarrow}
  S_{1,1}+S_{2,1} \]
\[F_2+S_{2,1} \overset{\tilde{a}_2}{\underset{\tilde{b}_2}{\rightleftarrows}}
  V_2 \overset{\tilde{c}_2}{\rightarrow}
  F_2+S_{2,0}.\]
\end{minipage}

The corresponding mass-action dynamical system is (see~\eqref{CRN}):
\[\begin{array}{rcl}
 \dot{s}_{1,0} & = &-a_1es_{1,0}+b_1u_1+\tilde{c}_1v_1,\\
 \dot{s}_{1,1} & = &c_1u_1-\tilde{a}_1f_1s_{1,1}+\tilde{b}_1v_1-a_2s_{1,1}s_{2,0}+(b_2+c_2)u_2,\\
 \dot{e}= -\dot{u_1} & = &-a_1es_{1,0}+(b_1+c_1)u_1,\\
 \dot{f}_1=-\dot{v}_1 & = &-\tilde{a}_1f_1s_{1,1}+(\tilde{b}_1+\tilde{c}_1)v_1,\\
 \dot{s}_{2,0}& = &-a_2s_{1,1}s_{2,0}+b_2u_2+\tilde{c}_2v_2,\\
 \dot{s}_{2,1}& = &c_2u_2-\tilde{a}_2f_2s_{2,1}+\tilde{b}_2v_2,\\
  \dot{u}_2& = &a_2s_{1,1}s_{2,0}-(b_2+c_2)u_2,\\
 \dot{f}_2=-\dot{v}_2 & = &-\tilde{a}_2f_2s_{2,1}+(\tilde{b}_2+\tilde{c}_2)v_2.

\end{array}\]
\end {example}

We prove in Theorem~\ref{thm:cascade} that all the parameters in a signaling cascade system can be identified from a single variable: the last product of the last layer ($S_{2,1}$ in the cascade presented in Example~\ref{2layer}). This species is usually an output of interest for this type of cascades \cite{AYK11,chen09,hagen13,LHG09}.

The organization of the paper is as follows. The next section provides introductory material on chemical reaction networks, mass-action kinetics equations and identifiability. Section~\ref{sec:assumptions} deals with the general assumptions required by the biochemical reaction networks we consider along the paper.
In Sections~\ref{sec:components} and~\ref{sec:cascade} we analyze the identifiability for sequential phosphorylation/dephosphorylation networks and phosphorylation cascades, respectively.
We illustrate these results in Section~\ref{sec:algorithm} with a procedure to determine, from (noise-free) data, the $30$ rate constants in the $3$-layer MAPK cascade, which relies on a heuristic to choose points to specialize the variables and solve for the rate constants.
Finally, we include an appendix with the complete proofs of the results stated in the paper.

\section{Preliminaries and Basic Notions} \label{sec:preliminaries}

\subsection{Chemical Reaction Systems} \label{subsec: chemical reaction}

We briefly recall the basic setup of chemical reaction networks and how they give rise to autonomous dynamical systems under mass-action kinetics.

Given a set of $s$ chemical species (denoted by capital letters), a {\em chemical reaction network} on this set of species is a finite directed graph whose vertices are indicated by complexes (non negative integer linear combinations of the species) and whose edges are labeled by parameters (positive reaction rate constants). The labeled digraph is denoted $G = (\Ve,\Edg, \Ka)$, with vertex set $\Ve$, edge set $\,\Edg$ and edge labels $\Ka \in \R_{>0}^{\#\Edg}$.
If $(y,y')\in \Edg$, we note $y\to y'$. The complexes determine  vectors in $\Z_{\ge 0}^{s}$ (the coefficients of the linear combinations) according to the stoichiometry of the species they consist of.
We identify each complex with its corresponding vector and also with the formal linear combination of species specified by its coordinates.

We present a basic example that illustrates how a chemical reaction network gives rise to a dynamical system. This example represents a classical mechanism of enzymatic reactions, usually known as the futile cycle \cite{sig-016,kholo00,ws08}:

\begin{example}\label{ex:enzsys}
Consider the following graph

\[ E+S_0 \overset{a}{\underset{b}{\rightleftarrows}} U \overset{c}{\rightarrow} E+S_1 \qquad
 F+S_1 \overset{\tilde{a}}{\underset{\tilde{b}}{\rightleftarrows}} V \overset{\tilde{c}}{\rightarrow} F+S_0.
 \]

The $s=6$ variables $U$, $V$, $S_0$, $S_1$, $E$, $F$, denote the chemical species. The source and the product of each reaction (i.e. the vertices) are the complexes (non negative linear combinations of the species). Finally the edge labels in $\Ka=(a,b,c,\tilde{a},\tilde{b},\tilde{c})$ are the reaction rate constants describing how concentrations of the six species change in time as the reactions occur.

The first three complexes give rise to
the vectors $(0,0,1,0,1,0)$, $(1,0,0,0,0,0)$ and $(0,0,0,1,1,0)$ while those in the second ones are $(0,0,0,1,0,1)$, $(0,1,0,0,0,0)$, and $(0,0,1,0,0,1)$.
\end{example}

A chemical reaction network $G$ as above, under the assumption of {\em mass-action kinetics} induces a polynomial dynamical system in the following way. Suppose that the species are $X_1,\ldots,X_s$
and their respective concentrations are denoted by $x_1,\ldots,x_s$ (denoted by small letters). We write $k_{yy'}$ for the reaction rate of each reaction $y\to y'$ in $\mathcal{R}$. We introduce the following chemical reaction dynamical system:
\begin{equation}
\label{CRN}
{\dot \bx}~=~\left( \frac{dx_1}{dt} ,\frac{dx_2}{dt}  ,\dots , \frac{dx_s}{dt}   \right)  ~=~
\underset{y\to y'}{\sum} k_{yy'} \,  \bx^y \, (y'-y),
\end{equation}
where $\bx:=(x_1,\dots,x_s)$ and $\bx^y:=x_1^{y_1}\cdots x_s^{y_s}$ if $y=(y_1,\ldots,y_s)$.
The right-hand side of each differential equation $\dot{x}_i$ is a polynomial $f_i(\bx,\Ka)$, in the variables $x_1,\dots, x_s$ with coefficients depending on the parameters $\Ka:=(k_{yy'})_{(y,y')\in \mathcal{R}}$.

For instance, in our previous Example \ref{ex:enzsys} this induced dynamical system is:

\begin{equation} \label{system enz}
\begin{array}{rcll}
\dot{u} & = & a es_0-(b+c)u,\\[1mm]
\dot{v} & = &\tilde{a} fs_1-(\tilde{b}+\tilde{c}) v,\\[1mm]
\dot{s}_0 & = &-a es_0+b u+\tilde{c} v,\\[1mm]
\dot{s}_1& = &-\tilde{a} fs_1+\tilde{b} v+c u,\\[1mm]
\dot{e} & = &-a es_0+(b+c) u,\\[1mm]
\dot{f} & = &-\tilde{a} fs_1+(\tilde{b}+\tilde{c}) v.
\end{array}
\end{equation}

\subsection{Identifiability in Chemical Reaction Systems}

Among all the different (not always equivalent) notions of identifiability in differential equations and control theory we have chosen to work from the one introduced in \cite{CP} since it seems specially well suited to the dynamical biochemical systems we consider here (see, for instance, \cite{chis11,Raue14} for a survey on the state of the art).

One of the main differences in the various approaches to identifiability is an assumption on the number of experiments that can be conducted with the same parameter values but different initial conditions: a single-experiment approach assumes the experiment is performed only once with some (often generic) initial condition (see, for example, \cite[Chapter 10]{DiStefano14}), whereas the multi-experiment approach we adopt in this paper assumes that it is allowed to perform as many experiments as needed with the same parameter values but different initial conditions.

\begin{definition} \label{ident}
Let $G = (\Ve,\Edg, \Ka)$ be a chemical reaction network with $s$ species. Its associated  reaction system \eqref{CRN} is called \emph{identifiable} if the map $\Phi:\R_{>0}^{\# \Edg}\to\R[\bx]^s$,
\[
\Phi(\Ka)=\underset{y\to y'}{\sum} k_{yy'}\bx^y(y'-y),
\]
is injective (here $\Ka=(k_{yy'})_{(y,y')\in\Edg}$ and $\R[\bx]$ is the polynomial ring in the variables $x_1,\ldots,x_s$).
\end{definition}

\begin{example}
In  Example \ref{ex:enzsys} (see the corresponding differential equation system (\ref{system enz})),
the domain of the map $\Phi$ is $\R_{>0}^6$, the target space is $\R[u,v,s_0,s_1,e,f]^6$ and the coordinate functions are the right-hand sides of the differential equations in (\ref{system enz}). It is clear that $\Phi$ is injective and therefore, the reaction system is identifiable: the right-hand sides of $\dot{s}_0$ and $\dot{s}_1$ determine the six constants $\Ka=(a,b,c,\tilde{a},\tilde{b},\tilde{c})$.
\end{example}

\begin{example}\label{ex:enzsys2} (see \cite[Section 2, Fig.1]{CP})
Consider the following graph
\[
\xymatrix{
2X_2\ar@{->}@<.5ex>[rr]^*-<1pt>{k_1} \ar@{->}@<.5ex>[ddr]^*-<1pt>{k_4} & & 2X_1 \ar@{->}@<.5ex>[ll]^*-<1pt>{k_6} \ar@{->}@<.5ex>[ddl]^*-<1pt>{k_2} \\ \\
& X_1+X_2 \ar@{->}@<.5ex>[uul]^*-<1pt>{k_3} \ar@{->}@<.5ex>[uur]^*-<1pt>{k_5}&
}.
\]
Here $s=2$, $\# \mathcal{R}=6$ and the associated dynamical system is
\begin{equation}\label{eq:reaction_CP}
\begin{array}{ccc}
\dot{x}_1&=&(2k_1+k_4)x_2^2-(k_2+2k_6)x_1^2+(k_5-k_3)x_1x_2\\[2mm]
\dot{x}_2&=&-(2k_1+k_4)x_2^2+(k_2+2k_6)x_1^2+(k_3-k_5)x_1x_2
\end{array}.
\end{equation}
Clearly, the map $\Phi$ is not injective: parameters $\Ka\in \R_{>0}^6$ define the same polynomials under $\Phi$ if and only if the linear forms $2k_1+k_4$, $k_2+2k_6$ and $k_5-k_3$ take the same values when evaluated at $\Ka$. For instance, $\Phi(1,1,1,1,1,1)=\Phi(1,1,2,1,2,1)=(3x_2^2-3x_1^2,-3x_2^2+3x_1^2)$.
Therefore, the system (\ref{eq:reaction_CP}) is not identifiable.
\end{example}

\begin{definition}
For a chemical reaction network $G$, we introduce the \emph{total derivative} (or \emph{Lie derivative}) associated to the induced differential equations system as follows: given a differentiable function $\varphi:\R^s\to \R$, its total derivative $\dot{\varphi}$ is defined as
\[
\dot{\varphi}:=\displaystyle{\sum_{i=1}^s  \dfrac{\partial \varphi}{\partial x_i} \dfrac{d x_i}{d t}}=\sum_{i=1}^s \dfrac{\partial \varphi}{\partial x_i} \sum_{y\to y'} k_{yy'}\bx ^y(y'_i-y_i),
\]
where each partial derivative $\dfrac{d x_i}{d t}$ is replaced according to system (\ref{CRN}).
For an integer $\ell \geq 1$, we denote by $\varphi^{(\ell)}$ the $\ell$-th iteration of the total derivative of $\varphi$ (in particular $\varphi^{(1)}=\dot{\varphi})$.

\end{definition}

For instance, for the network given in Example \ref{ex:enzsys}, its associated dynamical system (\ref{system enz}) and the function $\varphi=u^4+v$, we have
\[\dot{\varphi}=4u^3(a es_0-(b+c)u)+\tilde{a} fs_1-(\tilde{b}+\tilde{c})v.\]
Note that for a differentiable function $\varphi:\R^s\to \R$, the total derivative $\varphi^{(\ell)}$ can be regarded as a function depending on the $(s+\#\mathcal R)$-variables $\bx,\Ka$.

\begin{definition} \label{identvar}
Let $G = (\Ve,\Edg, \Ka)$ be a chemical reaction network with $s$ species. We say its associated  reaction system \eqref{CRN} is \emph{identifiable from the variables} $x_{i_1},\dots,x_{i_t}$ if there exists a positive integer $D$ such that the following injectivity condition holds: if $\Ka^{*},\Ka^{**}\in\R_{>0}^{\# \Edg}$ verify
\[ x^{(\ell)}_{i_j}(\bx,\Ka^{*})=x^{(\ell)}_{i_j}(\bx,\Ka^{**}),\]
for all $1\leq \ell\leq D$, $1\leq j\leq t$, then $\Ka^{*}=\Ka^{**}$.
\end{definition}

The introduction of the Lie derivative in identifiability is a usual and quite natural approach suitable adapted to our purposes (see, for instance, \cite{chis11}). Among other works following this approach,  \cite{Sed02}, \cite{CBB11} and \cite{AKJ12} also include a discussion about the number of derivatives needed for the proposed identifiability analysis.

Definitions \ref{ident} and \ref{identvar} are related in the obvious way:

\begin{proposition}
A chemical reaction system in the variables $\bx=x_1,\ldots, x_s$ is identifiable in the sense of Definition \ref{ident} if and only it is identifiable from the variables $x_1,\ldots, x_s$ in the sense of Definition \ref{identvar}.
\end{proposition}

\begin{proof}
First we observe that the identity  $\Phi=\dot{x}_1\times \dot{x}_2\times \cdots \times\dot{x}_s$ holds as functions of the argument $\Ka$. Thus, if $\Phi$ is injective, the condition of Definition \ref{identvar} is satisfied for the variables $x_1,\ldots, x_s$ and the integer $D=1$. Conversely, suppose that the chemical reaction system is identifiable from the variables $x_1,\ldots, x_s$ using a certain number $D$ of successive total derivatives. Then the function $\Phi$ is necessarily injective in the arguments $\Ka$: if it is not the case, there exist $\Ka^{*}\ne \Ka^{**}$ such that
$\dot{x}_i(\bx,\Ka^{*})=\dot{x}_i(\bx,\Ka^{**})$ as functions of the variables $\bx$ for all $i=1,\ldots,s$. Since the values of $\Ka^{*},\Ka^{**}$ are constants with respect to the total derivative we conclude that $x_i^{(\ell)}(\bx,\Ka^{*})=x_i^{(\ell)}(\bx,\Ka^{**})$ for all $\ell\in\N$ and all $1\le i \le s$, arriving at a contradiction.
\end{proof}

\begin{example}
Consider the graph
\[
X_1+X_2 \overset{k_1}{\longrightarrow} X_3 \overset{k_2}{\longrightarrow} X_4.
\]
and its associated system
\[
\dot{x}_1=-k_1x_1x_2,\quad \dot{x}_2=-k_1x_1x_2,\quad \dot{x}_3=k_1x_1x_2-k_2x_3,\quad \dot{x}_4=k_2x_3.
\]
The system is identifiable in the sense of Definition \ref{ident}. Following Definition \ref{identvar}, the system is identifiable from the single variable $x_3$ with one derivative (i.e. in this case $D=1$ in Definition \ref{identvar}). It is also identifiable from the variable $x_4$, but its total derivative of second order is needed in order to determine all the parameters (i.e. $D=2$ for this variable). On the other hand, the system is not identifiable from the set of variables $\{x_1,x_2\}$, since the constant $k_2$ does not appear in any of the successive total derivatives of $x_1$ nor $x_2$.
\end{example}

For technical reasons, we need to slightly generalize the notion of identifiability introduced in Definition \ref{identvar} above. The following definition is related to the notion of identifiability of parameter combinations \cite{Bo07,MEDS09}:

\begin{definition} \label{identvar2}
Let $G = (\Ve,\Edg, \Ka)$ be a chemical reaction network. Let $p\in \N$ and $\psi: \R_{>0}^{\# \Edg}\to \R^p$ be a map from the space of parameters in an affine space $\R^p$. We say that the map $\psi$ is \emph{identifiable from the variables} $x_{i_1},\dots,x_{i_t}$ if there exists a positive integer $D$ such that the following injectivity condition holds: if $\Ka^{*},\Ka^{**}\in\R_{>0}^{\# \Edg}$ verify
\[ x^{(\ell)}_{i_j}(\bx,\Ka^{*})=x^{(\ell)}_{i_j}(\bx,\Ka^{**}),\]
for all $1\leq \ell\leq D$, $1\leq j\leq t$, then $\psi(\Ka^{*})=\psi(\Ka^{**})$.
\end{definition}

Roughly speaking, Definition \ref{identvar2} says that the value of the function $\psi$ is uniquely determined by the values of the successive derivatives $x^{(\ell)}_{i_j}$.

Observe that the notion of identifiability of a system from the variables $x_{i_1},\ldots,x_{i_t}$ as it is defined in Definition \ref{identvar} can be translated in the sense of Definition \ref{identvar2} as the identifiability of the function $\psi:\R_{>0}^{\# \Edg}\to \R^{\# \Edg}$, $\psi(\Ka)=\Ka$.

For instance, in the (non identifiable) Example \ref{ex:enzsys2}, the function $\psi: \R_{>0}^6\to \R^3$, defined as $\psi(\Ka):=(2k_1+k_4,k_2+2k_6,k_5-k_3)$, is identifiable from $x_1$ (or $x_2$, or both variables). In this case we say simply that the constants $2k_1+k_4,\ k_2+2k_6,\ k_5-k_3$ can be identified from $x_1$.

This notion will be useful along the paper.
We will typically consider very simple functions $\psi$ whose coordinates are either the rate constants or the sum of all the rate constants leaving from one complex.

\section{Assumptions on the biochemical reaction networks} \label{sec:assumptions}

We will analyze the identifiability problem for a specific kind of chemical reaction networks. We start by describing the assumptions on the networks we will consider in the sequel.

First, we assume that the ``building blocks'' of the network have the following shape:
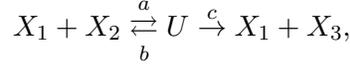
\[X_{1}+X_{2} \overset{a}{\underset{b}{\rightleftarrows}} U \overset{c}{\rightarrow} X_{1}+X_{3},\]
where $U$ is a species that only participates in those three reactions along all the network. We call $U$ an intermediate species and we say that species $X_1$ acts as an \emph{enzyme}, species $X_2$ acts as a \emph{substrate} and species $X_3$ acts as a \emph{product}.

\begin{definition}
 We say an intermediate species $U$ \emph{reacts to} the non-intermediate species $X_1$ if there exists another  non-intermediate species $X_2$ such that the reaction $U\to X_1+X_2$ exists. We say the non-intermediate
 species $X_1$ \emph{reacts with} the non-intermediate species $X_2$ if there exists an intermediate species
 $U$ such that the reaction $X_1+X_2\to U$ exists.
\end{definition}

\begin{example}[Example~\ref{2layer} continued]
Species $U_1,V_1,U_2,V_2$ are the intermediate species. $E$ and $F$ act as enzymes. $S_{1,0}$ acts as a substrate in the first connected component and as a product in the second one. Species $S_{2,0}$ and $S_{2,1}$ also act as both substrates and products (in the third and fourth connected components). Finally, $S_{1,1}$ acts as a product in the first connected component, as a substrate in the second one, and as an enzyme in the third one.
\end{example}

We make the following assumption concerning the structure of the network:
\begin{hypothesis}\label{hyp:network} \
\begin{enumerate}
\item Each connected component of the graph is of the following form:
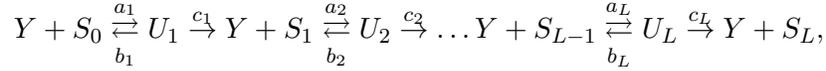
\[Y+S_0\overset{a_1}{\underset{b_1}{\rightleftarrows}} U_1 \overset{c_1}{\rightarrow}
Y+S_1\overset{a_2}{\underset{b_2}{\rightleftarrows}} U_2 \overset{c_2}{\rightarrow} \dots
Y+S_{L-1}\overset{a_L}{\underset{b_L}{\rightleftarrows}} U_L
\overset{c_{L}}{\rightarrow}Y+S_L,
\]
where there is a unique enzyme $Y$ acting on all the reactions of the connected component.
\item The intermediate species $U_j$ appearing in the entire network are all different.
\item The non-intermediate species $S_j$ in each connected component are all different, but they may also appear in other connected components.
\item Each complex lies in a unique connected component of the network.
\end{enumerate}
\end{hypothesis}

Although the above assumption seems restrictive, it is satisfied by many networks such as the multisite phosphorylation system described in Example~\ref{multisite}, the phosphorylation cascades as the one described in Example~\ref{2layer}, and also the network in Example~\ref{ex:enzsys}. As we observed before, in Examples~\ref{multisite} and~\ref{ex:enzsys} each species plays a unique role but in  Example \ref{2layer} the species $S_{1,1}$ acts alternatively as a product (in the first connected component), as a substrate (in the second one) and as an enzyme (in the third one).\\

For an intermediate species $U$, we call
\begin{align*}
  \Sp_U=\{S: S \text{ acts as a substrate or a product in the }\\ \nonumber
  \text{connected component determined by } U \}.
\end{align*}

For instance, in Example \ref{2layer} we have $\Sp_{U_1}=\Sp_{V_1}=\{S_{1,0},S_{1,1}\}$ and  $\Sp_{U_2}=\Sp_{V_2}=\{S_{2,0},S_{2,1}\}$.\\

We finish our assumptions on the kind of graphs we consider with a slightly technical condition.

\begin{hypothesis}\label{hyp:partition}
There is a partition of the species of the graph, that is, a decomposition into nonempty disjoint subsets:
 \begin{equation*}
\Sp=\Sp^{(0)} \bigsqcup \Sp^{(1)} \bigsqcup \dots \bigsqcup \Sp^{(M)},
\end{equation*}
where $M \ge 2$, $\bigsqcup$ denotes the disjoint union, $\Sp^{(0)}$ is the set of intermediate species and given an intermediate species $U$ with $Y$ acting as an enzyme
in the corresponding connected component, there exists $\alpha\geq 1$ with $\Sp_U\subseteq \Sp^{(\alpha)}$ and $Y\notin\Sp^{(\alpha)}$.
\end{hypothesis}

\begin{remark} \label{extra}
Under Assumption \ref{hyp:network}, the new condition imposed on the graph by Assumption \ref{hyp:partition} implies the following fact: if $X_1$ reacts with $X_2$, then there exists $\alpha\neq\beta$ such that $X_1\in\Sp^{(\alpha)}$ and $X_2\in\Sp^{(\beta)}$. 
In particular, if $S_i$ and $S_j$ are two substrates or products in the same connected component, the complex $S_i + S_j$ is not present in the network.
\end{remark}

\begin{example}
In Example \ref{2layer} we can consider the following partition $\Sp^{(0)}=\{U_1,V_1,U_2, V_2\}$, $\Sp^{(1)}=\{S_{1,0},S_{1,1}\}$, $\Sp^{(2)}=\{S_{2,0},S_{2,1}\}$, $\Sp^{(3)}=\{E\}$, $\Sp^{(4)}=\{F_1\}$, $\Sp^{(5)}=\{F_2\}$.

However it is not the unique possible partition: for instance, another choice could be
$\Sp^{(0)}, \Sp^{(1)}$ and $\Sp^{(2)}$ as before, but $\Sp^{(3)}$, $\Sp^{(4)}$ and $\Sp^{(5)}$ are replaced by the single set $\{E,F_1,F_2\}$.
\end{example}

\section{Identifiability in connected components}\label{sec:components}

This section is devoted to dealing with the identifiability problem for chemical reaction networks satisfying the assumptions stated in Section \ref{sec:assumptions}. Our aim is to show that all reaction constants of the network can be identified from the successive derivatives of the variables in a certain family of non-intermediates.

In order to do this, we choose a suitable subset of variables and estimate the maximum number of successive derivatives of them that we need to identify all the reaction constants. Namely, we choose variables $x_{i_1},\dots, x_{i_t}$ and determine a number $D_j$ of successive derivatives of $x_{i_j}$, for $1\le j\le t$, so that the injectivity condition in Definition \ref{identvar} holds for $D = \max\{D_j\}$.

Since the derivatives $x_{i_j}^{(\ell)}(\bx, \Ka)$  are polynomials in the variables $\bx$ with coefficients that are polynomials in the reaction rate constants $\Ka$, showing that the parameters $\Ka$ are identifiable from $x_{i_j}^{(\ell)}(\bx, \Ka)$ for $1\le \ell \le D_j$, $1\le j \le t$, is the same as showing that they are uniquely determined by the coefficients of the polynomials $x_{i_j}^{(\ell)}(\bx, \Ka)$.  Thus, our strategy to proving identifiability will be to locate suitable subsets of monomials in the derivatives $x_{i_j}^{(\ell)}$ that enable us to prove that the values of all the reaction constants can be uniquely determined from  their corresponding coefficients.

\subsection{Identifying the constants in one connected component from one variable}\label{subsec:1comp}

The aim of this section is to show that all the reaction constants in a connected component
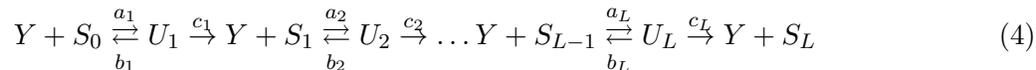
\begin{equation}\label{eq:conncomp}
Y+S_0\overset{a_1}{\underset{b_1}{\rightleftarrows}} U_1 \overset{c_1}{\rightarrow}
Y+S_1\overset{a_2}{\underset{b_2}{\rightleftarrows}} U_2 \overset{c_2}{\rightarrow} \dots
Y+S_{L-1}\overset{a_L}{\underset{b_L}{\rightleftarrows}} U_L
\overset{c_{L}}{\rightarrow}Y+S_L
\end{equation}
of a network satisfying the assumptions stated in Section \ref{sec:assumptions}
are identifiable from a limited number of successive derivatives of the variable $s_L$ representing the concentration of the last product.

We start by showing that all the constants $c_L, a_L, b_L$,  and, for $1\le j \le L-1$, $a_j$ and $b_j+c_j$ can be identified (in the sense of Definition \ref{identvar2}) simply from the first three derivatives of this variable. Then, we proceed to identify recursively all the constants $c_{j}$ (and consequently, also the constants $b_j$) for $j=L-1,\dots, 1$, from higher order derivatives of $s_L$. The main result of this section is the following:

\begin{proposition} \label{prop:1conncomp}
 All the constants in a connected component
 \[Y+S_0\overset{a_1}{\underset{b_1}{\rightleftarrows}} U_1 \overset{c_1}{\rightarrow}
Y+S_1\overset{a_2}{\underset{b_2}{\rightleftarrows}} U_2 \overset{c_2}{\rightarrow} \dots
Y+S_{L-1}\overset{a_L}{\underset{b_L}{\rightleftarrows}} U_L
\overset{c_{L}}{\rightarrow}Y+S_{L}\]
of a network satisfying the assumptions in Section \ref{sec:assumptions}
can be identified from $s^{(\ell)}_{L}$ with $1\leq \ell \leq \mathrm{max}\{2,2L-1\}$.
\end{proposition}

The strategy in the proof of this result consists in the exact computation of the coefficients of certain distinguished monomials in the successive derivatives of $s_L$. This explicit computation enables us to achieve the identifiability of all the constants of the connected component by means of a recursive procedure, that we summarize in Table~\ref{table:1component}.  For a complete proof, see Proposition \ref{propApp:1conncomp} in Appendix \ref{sec:pf}.

\begin{center}
\begin{table}[h]
\renewcommand{\arraystretch}{1.5}
\begin{tabular}{|c|c|c|c|}
\hline
Derivative   & Monomial          & Coefficient (up to sign)  & Identified constant       \\ \hline
$\dot{s}_L$ & $u_L$             & $c_L$              & $c_L$         \\  \hline
\multirow{2}{*}
{$\ddot{s}_L$} & $y s_{L-1}$       & $c_L a_L$          & $a_L$         \\  \cline{2-4}
               & $u_L$             & $c_L (b_L+c_L)$    & $b_L$         \\   \hline
\multirow{2}{*}
{$s_L^{(3)}$} & $y s_{j-1}s_{L-1}$ & $c_L a_L a_j$     & $a_j$  {\small $(1\le j \le L-1)$}\\ \cline{2-4}
              & $u_j s_{L-1}$      & $c_L a_L (b_j+c_j)$     & $b_j+c_j$ {\small $(1\le j \le L-1)$}\\ \hline
{$s_L^{(2k+1)}$} & $y^k u_{L-k}$ & $c_{L-k} \prod_{j=0}^{k-1} c_{L-j}a_{L-j}$ & $c_{L-k}$, $b_{L-k}$ {\small $(1\le k \le L-1)$}\\ \hline
\end{tabular}

\caption{The constants in the connected component~\eqref{eq:conncomp} can be identified from
$s_L$. The table shows the monomials to be considered (column 2) in each of the successive derivatives of $s_L$ (column 1). For each monomial, taking into account the constants already identified, the corresponding coefficient (column 3) enables us to identify the constant appearing in the last column.}
\label{table:1component}
\end{table}
\end{center}

We illustrate the procedure underlying the proof of the previous statement with a simple example.

\begin{example}
Consider the network
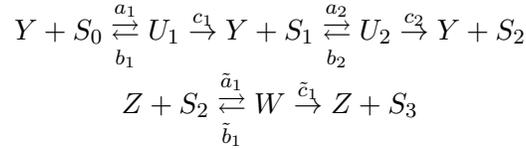
\[\begin{array}{c}
Y+S_0\overset{a_1}{\underset{b_1}{\rightleftarrows}} U_1 \overset{c_1}{\rightarrow}
Y+S_1\overset{a_2}{\underset{b_2}{\rightleftarrows}} U_2 \overset{c_2}{\rightarrow} Y+S_2\\
Z+S_2\overset{\tilde{a}_1}{\underset{\tilde{b}_1}{\rightleftarrows}} W \overset{\tilde{c}_1}{\rightarrow}Z+S_3
\end{array}
\]
According to Proposition~\ref{prop:1conncomp}, all the constants in the first connected component can be identified from $s^{(\ell)}_2$ with $1\leq \ell \leq 3$.
In fact, if we call $K_1=b_1+c_1$, $K_2=b_2+c_2$ and $\tilde{K}_1=\tilde{b}_1+\tilde{c}_1$,
\begin{align*}
 \dot{s}_2 =&-\tilde{a}_1s_2z+\fboxrule=1.2pt\fcolorbox{blue}{white}{$c_2$}u_2+\tilde{b}_1w,\\
 \ddot{s}_2=& -\tilde{a}_1[\dot{s}_2z+s_2(\underset{\dot{z}}{\underbrace{-\tilde{a}_1s_2z+\tilde{K}_1w}})]+
  c_2[\underset{\dot{u}_2}{\underbrace{\fboxrule=1.2pt\fcolorbox{blue}{white}{$a_2$}s_1y-\fboxrule=1.2pt\fcolorbox{blue}{white}{$K_2$}u_2}}]+
  \tilde{b}_1[\underset{\dot{w}}{\underbrace{\tilde{a}_1s_2z-\tilde{K}_1w}}],\\
  s^{(3)}_2=&-\tilde{a}_1[\ddot{s}_2z+2\dot{s}_2\dot{z}+s_2(-\tilde{a}_1(\dot{s}_2z+s_2\dot{z})+\tilde{K}_1\dot{w})]+\\
  &+c_2[a_2(\underset{\dot{s}_1}{(\underbrace{-a_2s_1y+\fboxrule=1.2pt\fcolorbox{blue}{white}{$c_1$}u_1+b_2u_2})}y+
  s_1\underset{\dot{y}}{(\underbrace{-\fboxrule=1.2pt\fcolorbox{blue}{white}{$a_1$}s_0y+\fboxrule=1.2pt\fcolorbox{blue}{white}{$K_1$}u_1-a_2s_2y+K_2u_2})})-K_2\dot{u}_2]+\\
  &+\tilde{b}_1[\tilde{a}_1(\dot{s}_2z+s_2\dot{z})-\tilde{K}_1\dot{w}],
\end{align*}
where the constants $c_2,a_2,K_2$ (thus, also $b_2= K_2-c_2$), $a_1,K_1$ and $c_1$ (thus, also $b_1= K_1 - c_1$) are identified following Table \ref{table:1component}.
\end{example}

A direct consequence of Proposition \ref{prop:1conncomp} is the following theorem:

\begin{theorem} \label{thm:eachconncomp}
If a chemical reaction network satisfying the assumptions in Section \ref{sec:assumptions} consists of $N$ connected components
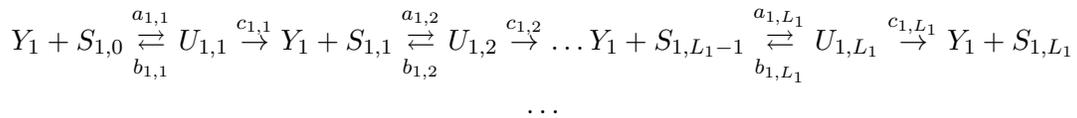
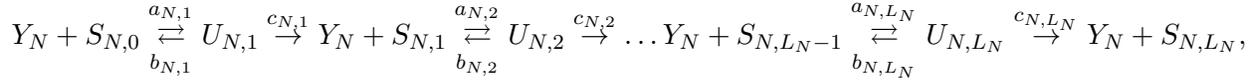
\[Y_1+S_{1,0}\overset{a_{1,1}}{\underset{b_{1,1}}{\rightleftarrows}} U_{1,1} \overset{c_{1,1}}{\rightarrow}
Y_1+S_{1,1}\overset{a_{1,2}}{\underset{b_{1,2}}{\rightleftarrows}} U_{1,2} \overset{c_{1,2}}{\rightarrow} \dots
Y_1+S_{1,L_1-1}\overset{a_{1,L_1}}{\underset{b_{1,L_1}}{\rightleftarrows}} U_{1,L_1}
\overset{c_{1,L_1}}{\rightarrow}Y_1+S_{1,L_1}\]
\[\dots\]
\[Y_N+S_{N,0}\overset{a_{N,1}}{\underset{b_{N,1}}{\rightleftarrows}} U_{N,1} \overset{c_{N,1}}{\rightarrow}
Y_N+S_{N,1}\overset{a_{N,2}}{\underset{b_{N,2}}{\rightleftarrows}} U_{N,2} \overset{c_{N,2}}{\rightarrow} \dots
Y_N+S_{N,L_N-1}\overset{a_{N,L_N}}{\underset{b_{N,L_N}}{\rightleftarrows}} U_{N,L_N}\overset{c_{N,L_N}}{\rightarrow}Y_N+S_{N,L_N},\]
then the associated system is identifiable from the variables $s_{1,L_1},\dots,s_{N,L_N}$ corresponding to the last products of each connected component of the network.
Moreover, for every $1\le i \le N$, the order of derivation needed for the variable $s_{i,L_i}$ is at most  $\max\{2,2L_i-1\}$.
\end{theorem}

\subsection{Identifying the constants in two connected components from one variable}\label{subsec:2comp}

In this subsection we analyze the identifiability problem for a subclass of the networks we have been considering.
More precisely, we consider networks containing pairs of connected components of the following type:
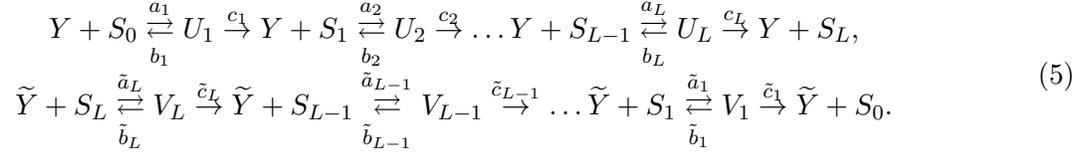
\begin{equation}\label{eq:2conncomp}
 \begin{array}{c}
  Y+S_0\overset{a_1}{\underset{b_1}{\rightleftarrows}} U_1 \overset{c_1}{\rightarrow}
  Y+S_1\overset{a_2}{\underset{b_2}{\rightleftarrows}} U_2 \overset{c_2}{\rightarrow} \dots
  Y+S_{L-1}\overset{a_L}{\underset{b_L}{\rightleftarrows}} U_L
  \overset{c_{L}}{\rightarrow}Y+S_L,\\
  \widetilde Y +S_L\overset{\tilde{a}_L}{\underset{\tilde{b}_L}{\rightleftarrows}} V_L \overset{\tilde{c}_L}{\rightarrow}
  \widetilde Y+S_{L-1}\overset{\tilde{a}_{L-1}}{\underset{\tilde{b}_{L-1}}{\rightleftarrows}} V_{L-1} \overset{\tilde{c}_{L-1}}{\rightarrow} \dots
  \widetilde Y+S_1\overset{\tilde{a}_1}{\underset{\tilde{b}_1}{\rightleftarrows}} V_1
  \overset{\tilde{c}_1}{\rightarrow}\widetilde Y+S_0.
 \end{array}
 \end{equation}

As before, we work under the assumptions made in Section~\ref{sec:assumptions}. 

By Proposition \ref{prop:1conncomp}, we know that all the constants in the first connected component in \eqref{eq:2conncomp} can be identified from a certain number of successive derivatives of $s_L$. Using the specific structure of the second component, we can prove that the same derivatives also enable the identification of the reaction rate constants of that component.

We first prove that the constants $\tilde{a}_L, \tilde{b}_L, \tilde{c}_L$, and, for $1\le j \le L-1$, $\tilde{a}_j$ and $\tilde{b}_j + \tilde{c}_j$ can be identified from $\dot{s}_L$ and $\ddot{s}_L$ and, then, by means of a recursive explicit computation of coefficients of a family of distinguished monomials in higher order derivatives of $s_L$, we show how to successively identify the constants $\tilde{b}_j$ for $j={L-1},\dots, 1$, and, consequently, also the constants $\tilde{c}_j$.
In this way, we deduce:

\begin{proposition}\label{prop:2comp}
Given a chemical reaction network satisfying the assumptions in Section~\ref{sec:assumptions}, all the constants in two connected components of the type
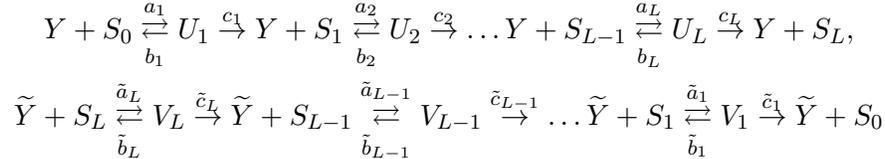
\[
\begin{array}{c}
Y+S_0\overset{a_1}{\underset{b_1}{\rightleftarrows}} U_1 \overset{c_1}{\rightarrow}
  Y+S_1\overset{a_2}{\underset{b_2}{\rightleftarrows}} U_2 \overset{c_2}{\rightarrow} \dots
  Y+S_{L-1}\overset{a_L}{\underset{b_L}{\rightleftarrows}} U_L
  \overset{c_{L}}{\rightarrow}Y+S_{L},\\[4mm]
\widetilde{Y}+S_{L}\overset{\tilde{a}_L}{\underset{\tilde{b}_L}{\rightleftarrows}} V_L \overset{\tilde{c}_L}{\rightarrow}
  \widetilde{Y}+S_{L-1}\overset{\tilde{a}_{L-1}}{\underset{\tilde{b}_{L-1}}{\rightleftarrows}} V_{L-1} \overset{\tilde{c}_{L-1}}{\rightarrow} \dots
  \widetilde{Y}+S_1\overset{\tilde{a}_1}{\underset{\tilde{b}_1}{\rightleftarrows}} V_1
  \overset{\tilde{c}_1}{\rightarrow}\widetilde{Y}+S_0
\end{array}
\]
can be identified from $s^{(\ell)}_{L}$ with $1\leq \ell \leq \mathrm{max}\{2,2L-1\}$.
\end{proposition}

We summarize the identifiability procedure underlying the proof of the previous proposition in  Table~\ref{table:2components}, and we also illustrate the result in Example~\ref{ex:2components}. For a complete proof, see Proposition \ref{propApp:2comp} in Appendix \ref{sec:pf}.

\begin{center}
\begin{table}[h]
 \renewcommand{\arraystretch}{1.5}

\begin{tabular}{|c|c|c|c|}
\hline
Derivative   & Monomial          & Coefficient (up to sign)  & Identified constant       \\ \hline
\multirow{3}{*}
{$\dot{s}_L$} & $u_L$             & $c_L$              & $c_L$         \\  \cline{2-4}
              & $\tilde y s_L$    & $\tilde{a}_L$      & $\tilde{a}_L$ \\  \cline{2-4}
              &   $v_L$           & $\tilde{b}_L$      & $\tilde{b}_L$ \\  \hline
\multirow{5}{*}
{$\ddot{s}_L$} & $y s_{L-1}$       & $c_L a_L$          & $a_L$         \\  \cline{2-4}
               & $u_L$             & $c_L (b_L+c_L)$          & $b_L$         \\  \cline{2-4}
               & $v_L$             & $\tilde{b}_L(\tilde{b}_L+\tilde{c}_L)$  & $\tilde{c}_L$ \\  \cline{2-4}
               & $\tilde{y}s_js_L$ & $\tilde{a}_L \tilde{a}_j$ & $\tilde{a}_j$ {\small $(1\le j \le L-1)$} \\  \cline{2-4}
               & $v_j s_L$         & $\tilde{a}_L (\tilde{b}_j+\tilde{c}_j)$ & $\tilde{b}_j+\tilde{c}_j$ {\small $(1\le j \le L-1)$} \\  \hline
\multirow{2}{*}
{$s_L^{(3)}$} & $y s_{j-1}s_{L-1}$ & $c_L a_L a_j$     & $a_j$  {\small $(1\le j \le L-1)$}\\ \cline{2-4}
              & $u_j s_{L-1}$      & $c_L a_L (b_j+c_j)$     & $b_j+c_j$ {\small $(1\le j \le L-1)$}\\ \hline
{$s_L^{(2k+1)}$} & $y^k u_{L-k}$ & $c_{L-k} \prod_{j=0}^{k-1} c_{L-j}a_{L-j}$ & $c_{L-k}$, $b_{L-k}$ \\ \cline{2-4}
{\small $k=1,\dots, L-1$}& $y^k v_{L-k}$ & $\tilde{b}_{L-k} \prod_{j=0}^{k-1} c_{L-j}a_{L-j}$ & $\tilde{b}_{L-k}$, $\tilde{c}_{L-k}$ \\ \hline
\end{tabular}
\caption{The constants in the two connected components in~\eqref{eq:2conncomp} can be identified from
$s_L$. Tha table shows the monomials to be considered (column 2) in each of the successive derivatives of $s_L$ (column 1). For each monomial, taking into account the constants already identified, the corresponding coefficient (column 3) enables us to identify the constant appearing in the last column.}\label{table:2components}
\end{table}
\end{center}

\begin{example}\label{ex:2components}
Consider the network
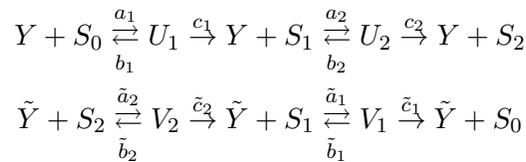
\[
\begin{array}{c}
Y+S_0\overset{a_1}{\underset{b_1}{\rightleftarrows}} U_1 \overset{c_1}{\rightarrow}
Y+S_1\overset{a_2}{\underset{b_2}{\rightleftarrows}} U_2 \overset{c_2}{\rightarrow} Y+S_2\\[4mm]
\tilde{Y}+S_2\overset{\tilde{a}_2}{\underset{\tilde{b}_2}{\rightleftarrows}} V_2 \overset{\tilde{c}_2}{\rightarrow}
\tilde{Y}+S_1\overset{\tilde{a}_1}{\underset{\tilde{b}_1}{\rightleftarrows}} V_1 \overset{\tilde{c}_1}{\rightarrow} \tilde{Y}+S_0
\end{array}
\]

According to Proposition~\ref{prop:2comp}, all the constants in the two connected components can be identified from $s^{(\ell)}_2$ with $1\leq \ell \leq 3$.
In fact, if we call $K_1=b_1+c_1$, $K_2=b_2+c_2$,  $\tilde{K}_1=\tilde{b}_1+\tilde{c}_1$ and $\tilde{K}_2=\tilde{b}_2+\tilde{c}_2$:

\begin{eqnarray*}
\dot{s}_2 &=& -\fboxrule=1.2pt\fcolorbox{blue}{white}{$\tilde{a}_2$}s_2\tilde{y}+
\fboxrule=1.2pt\fcolorbox{blue}{white}{$c_2$}u_2+\fboxrule=1.2pt\fcolorbox{blue}{white}{$\tilde{b}_2$}v_2,\\
\ddot{s}_2&=& -\tilde{a}_2[\dot{s}_2\tilde{y}+s_2(\underset{\dot{\tilde{y}}}{\underbrace{-\fboxrule=1.2pt\fcolorbox{blue}{white}{$\tilde{a}_1$}s_1\tilde{y}+\fboxrule=1.2pt\fcolorbox{blue}{white}{$\tilde{K}_1$}v_1-\tilde{a}_2s_2\tilde{y}+\fboxrule=1.2pt\fcolorbox{blue}{white}{$\tilde{K}_2$}v_2}})] +c_2[\underset{\dot{u}_2}{\underbrace{\fboxrule=1.2pt\fcolorbox{blue}{white}{$a_2$}s_1y-\fboxrule=1.2pt\fcolorbox{blue}{white}{$K_2$}u_2}}]+
  \tilde{b}_2[\underset{\dot{v_2}}{\underbrace{\tilde{a}_2s_2\tilde{y}-\tilde{K}_2v_2}}],\\
s^{(3)}_2&=&-\tilde{a}_2[\ddot{s}_2\tilde{y}+2\dot{s}_2\dot{\tilde{y}}+
s_2(-\tilde{a}_1(\dot{s}_1\tilde{y} +s_1\dot{\tilde{y}})
+\tilde{K}_1(\tilde{a}_1s_1\tilde{y}-\tilde{K}_1v_1)-\tilde{a}_2(\dot{s}_2\tilde{y}+s_2\dot{\tilde{y}})+\tilde{K}_2\dot{v}_2)]+\\
 & & +c_2[a_2((\underset{\dot{s}_1}{\underbrace{-a_2s_1y-\tilde{a}_1s_1\tilde{y}+
  \fboxrule=1.2pt\fcolorbox{blue}{white}{$c_1$}u_1+b_2u_2+\tilde{c}_2v_2+\fboxrule=1.2pt\fcolorbox{blue}{white}{$\tilde{b}_1$}v_1}})y+ \\
&&+s_1(\underset{\dot{y}}{\underbrace{-\fboxrule=1.2pt\fcolorbox{blue}{white}{$a_1$}s_0y+\fboxrule=1.2pt\fcolorbox{blue}{white}{$K_1$}u_1-a_2s_2y+K_2u_2}}))-K_2\dot{u}_2]+\tilde{b}_2[\tilde{a}_2(\dot{s}_2\tilde{y}+s_2\dot{\tilde{y}})-\tilde{K}_2\dot{v}_2].
\end{eqnarray*}

Here, the constants $c_2$, $\tilde{a}_2,  \tilde{b}_2$, $a_2, K_2$ (then, $b_2$), $\tilde{c}_2 $, $\tilde{a}_1, \tilde{K}_1$, $a_1$, $K_1$, $c_1$ (then, $b_1$), and $\tilde{b}_1$ (then, $\tilde{c}_1$) are identified following Table \ref{table:2components}.
\end{example}

A direct consequence of Proposition \ref{prop:2comp} is the following corollary:

\begin{corollary} \label{coro:2comp}
If a chemical reaction network satisfying the assumptions in Section \ref{sec:assumptions} consists of $2N$ connected components of the shape
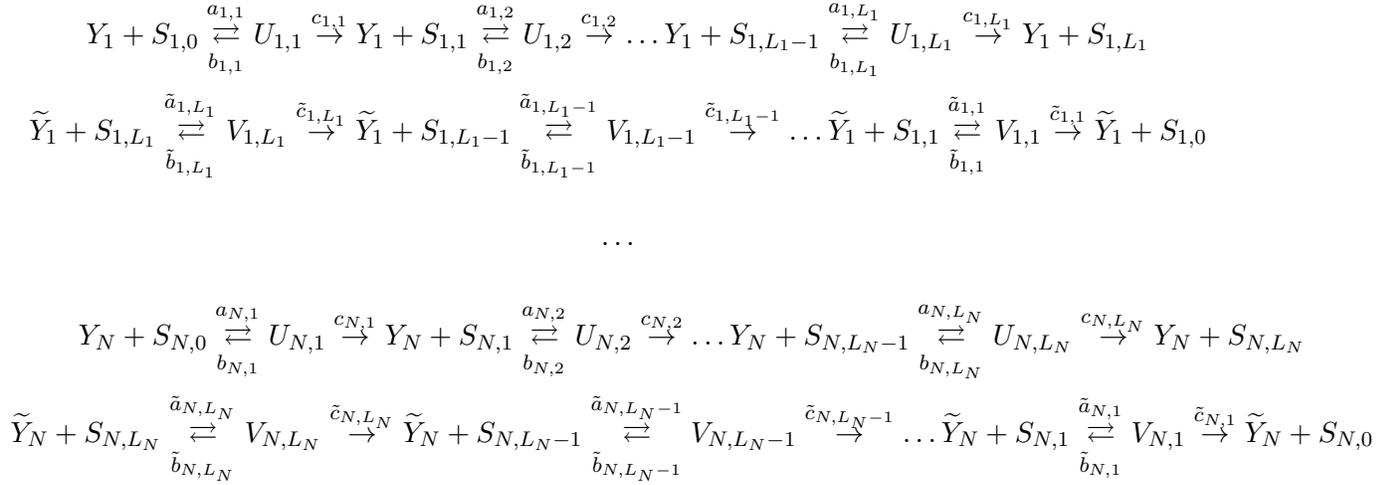
\[
\begin{array}{c}
Y_1+S_{1,0}\overset{a_{1,1}}{\underset{b_{1,1}}{\rightleftarrows}} U_{1,1} \overset{c_{1,1}}{\rightarrow}
Y_1+S_{1,1}\overset{a_{1,2}}{\underset{b_{1,2}}{\rightleftarrows}} U_{1,2} \overset{c_{1,2}}{\rightarrow} \dots
Y_1+S_{1,L_1-1}\overset{a_{1,L_1}}{\underset{b_{1,L_1}}{\rightleftarrows}} U_{1,L_1}
\overset{c_{1,L_1}}{\rightarrow}Y_1+S_{1,L_1}\\[5mm]
\widetilde{Y}_1+S_{1,L_1}\overset{\tilde{a}_{1,L_1}}{\underset{\tilde{b}_{1,L_1}}{\rightleftarrows}} V_{1,L_1} \overset{\tilde{c}_{1,L_1}}{\rightarrow}
\widetilde{Y}_1+S_{1,L_1-1}\overset{\tilde{a}_{1,L_1-1}}{\underset{\tilde{b}_{1,L_1-1}}{\rightleftarrows}} V_{1,L_1-1} \overset{\tilde{c}_{1,L_1-1}}{\rightarrow} \dots
\widetilde{Y}_1+S_{1,1}\overset{\tilde{a}_{1,1}}{\underset{\tilde{b}_{1,1}}{\rightleftarrows}} V_{1,1}
\overset{\tilde{c}_{1,1}}{\rightarrow}\widetilde{Y}_1+S_{1,0}
\end{array}
\]
\[\dots\]
\[\begin{array}{c}
Y_N+S_{N,0}\overset{a_{N,1}}{\underset{b_{N,1}}{\rightleftarrows}} U_{N,1} \overset{c_{N,1}}{\rightarrow}
Y_N+S_{N,1}\overset{a_{N,2}}{\underset{b_{N,2}}{\rightleftarrows}} U_{N,2} \overset{c_{N,2}}{\rightarrow} \dots
Y_N+S_{N,L_N-1}\overset{a_{N,L_N}}{\underset{b_{N,L_N}}{\rightleftarrows}} U_{N,L_N}
\overset{c_{N,L_N}}{\rightarrow}Y_N+S_{N,L_N}\\[5mm]
\widetilde{Y}_N+S_{N,L_N}\overset{\tilde{a}_{N,L_N}}{\underset{\tilde{b}_{N,L_N}}{\rightleftarrows}} V_{N,L_N} \overset{\tilde{c}_{N,L_N}}{\rightarrow}
  \widetilde{Y}_N+S_{N,L_N-1}\overset{\tilde{a}_{N,L_N-1}}{\underset{\tilde{b}_{N,L_N-1}}{\rightleftarrows}} V_{N,L_N-1} \overset{\tilde{c}_{N,L_N-1}}{\rightarrow} \dots
  \widetilde{Y}_N+S_{N,1}\overset{\tilde{a}_{N,1}}{\underset{\tilde{b}_{N,1}}{\rightleftarrows}} V_{N,1}
  \overset{\tilde{c}_{N,1}}{\rightarrow}\widetilde{Y}_N+S_{N,0}
\end{array}
\]
then the associated system is identifiable from the variables $s_{1,L_1},\dots,s_{N,L_N}$.
Moreover, for every $1\le i \le N$, the order of derivation needed for the variable $s_{i,L_i}$ is at most  $\max\{2,2L_i-1\}$.
\end{corollary}

\section{Identifying the cascade} \label{sec:cascade}

We will consider in this section networks that are called \emph{cascades}. Signaling cascades are biochemical networks of major biological importance as they participate in a number of several diseases and also control key cellular functions \cite{davis,kyriakis,pearson,schaeffer,widmann,zarubin}.  The Mitogen-Activated Protein Kinase (MAPK) cascade is a network present in all eukaryotic cells and one of the most extensively modeled signaling systems \cite{hornberg,sig-016,qiao07}. A schematic representation of the network is the following

{\small
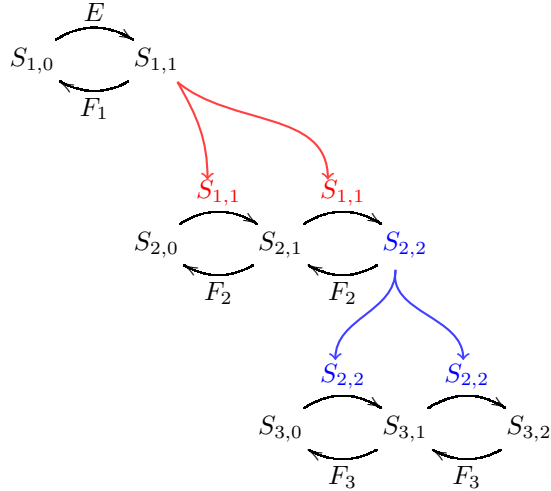
\begin{minipage}{0.99\textwidth}
\[
\xymatrix{
S_{1,0}\ar@{->}@/^1pc/[r]^*-<1pt>{E} &
S_{1,1} \ar@{->}@/^1pc/[l]^*-<1pt>{F_1}& & \\ \\
& S_{2,0}\ar@{->}@/^1pc/[r]^*-<1pt>{\textcolor{red}{S_{1,1}}} &
S_{2,1} \ar@{->}@/^1pc/[r]^*-<1pt>{\textcolor{red}{S_{1,1}}} \ar@{->}@/^1pc/[l]^*-<1pt>{F_2}&
\textcolor{blue}{S_{2,2}}\ar@{->}@/^1pc/[l]^*-<1pt>{F_2} & \\ \\
& &  S_{3,0}\ar@{->}@/^1pc/[r]^*-<1pt>{\textcolor{blue}{S_{2,2}}} &
S_{3,1} \ar@{->}@/^1pc/[r]^*-<1pt>{\textcolor{blue}{S_{2,2}}} \ar@{->}@/^1pc/[l]^*-<1pt>{F_3}&
S_{3,2}\ar@{->}@/^1pc/[l]^*-<1pt>{F_3} \\
}
\]
\begin{tikzpicture}[remember picture, overlay]
  \draw[->,color=red!75,thick] (6.8,5.6) to [out=-60,in=90] node[] {} (7.2,4.3) ;
  \draw[->,color=red!75,thick] (6.8,5.6) to [out=-40,in=90] node[] {} (8.8,4.3) ;
  \draw[->,color=blue!75,thick] (9.7,3.1) to [out=-90,in=90] node[] {} (8.9,1.9) ;
  \draw[->,color=blue!75,thick] (9.7,3.1) to [out=-100,in=90] node[] {} (10.6,1.9) ;
\end{tikzpicture}
\end{minipage}
}

\noindent where $S_{1,0}$ represents the kinase {\small MAPKKK}, and $S_{1,1}$ represents the activated form {\small MAPKKK$^*$}. $S_{2,0}$, $S_{2,1}$ and $S_{2,2}$ stand for {\small MAPKK}, {\small MAPKK-P} and {\small MAPKK-PP}, respectively. And finally, $S_{3,0}$, $S_{3,1}$ and $S_{3,2}$ stand for {\small MAPK}, {\small MAPK-P} and {\small MAPK-PP}, respectively. $F_1$ represents the enzyme that deactivates {\small MAPKKK$^*$}, and $F_2$ and $F_3$ represent the corresponding phosphatase of each layer.

\medskip

More generally, cascades consist of  $N\geq 1$ layers and are represented by the following scheme:

\medskip

{\small
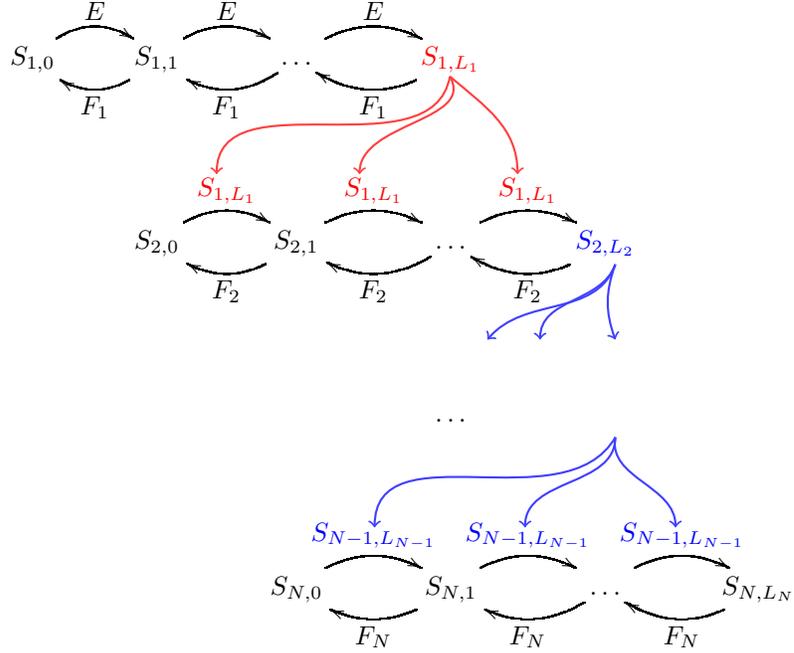
\begin{minipage}{0.99\textwidth}
\[
\xymatrix{
S_{1,0}\ar@{->}@/^1pc/[r]^*-<1pt>{E} & S_{1,1} \ar@{->}@/^1pc/[r]^*-<1pt>{E}\ar@{->}@/^1pc/[l]^*-<1pt>{F_1}&
{\phantom{S}}\dots{\phantom{S}}\ar@{->}@/^1pc/[r]^*-<1pt>{E}\ar@{->}@/^1pc/[l]^*-<1pt>{F_1}
& \textcolor{red}{S_{1,L_1}}\ar@{->}@/^1pc/[l]^*-<1pt>{F_1} & & \\ \\
& S_{2,0}\ar@{->}@/^1pc/[r]^*-<1pt>{\textcolor{red}{S_{1,L_1}}} &  S_{2,1} \ar@{->}@/^1pc/[r]^*-<1pt>{\textcolor{red}{S_{1,L_1}}} \ar@{->}@/^1pc/[l]^*-<1pt>{F_2}&
{\phantom{S}}\dots{\phantom{S}}\ar@{->}@/^1pc/[r]^*-<1pt>{\textcolor{red}{S_{1,L_1}}} \ar@{->}@/^1pc/[l]^*-<1pt>{F_2}
&  \textcolor{blue}{S_{2,L_2}}\ar@{->}@/^1pc/[l]^*-<1pt>{F_2} & \\ \\
&  &  & \dots &  & \\ \\
& &  S_{N,0}\ar@{->}@/^1pc/[r]^*-<1pt>{\textcolor{blue}{S_{N-1,L_{N-1}}}} &  S_{N,1} \ar@{->}@/^1pc/[r]^*-<1pt>{\textcolor{blue}{S_{N-1,L_{N-1}}}} \ar@{->}@/^1pc/[l]^*-<1pt>{F_N}&
{\phantom{S}}\dots{\phantom{S}}\ar@{->}@/^1pc/[r]^*-<1pt>{\textcolor{blue}{S_{N-1,L_{N-1}}}} \ar@{->}@/^1pc/[l]^*-<1pt>{F_N}
&  S_{N,L_N}\ar@{->}@/^1pc/[l]^*-<1pt>{F_N} \\
}
\]
\begin{tikzpicture}[remember picture, overlay]
  \draw[->,color=red!75,thick] (8.8,7.8) to [out=-100,in=90] node[] {} (5.7,6.5) ;
  \draw[->,color=red!75,thick] (8.8,7.8) to [out=-60,in=90] node[] {} (7.6,6.5) ;
  \draw[->,color=red!75,thick] (8.8,7.8) to [out=-40,in=90] node[] {} (9.7,6.5) ;
  \draw[->,color=blue!75,thick] (11,5.3) to [out=-110,in=50] node[] {} (9.3,4.3) ;
  \draw[->,color=blue!75,thick] (11,5.3) to [out=-110,in=90] node[] {} (10,4.3) ;
  \draw[->,color=blue!75,thick] (11,5.3) to [out=-110,in=110] node[] {} (11,4.3) ;
  \draw[->,color=blue!75,thick] (11,3) to [out=-120,in=90] node[] {} (7.8,1.8) ;
  \draw[->,color=blue!75,thick] (11,3) to [out=-90,in=90] node[] {} (9.8,1.8) ;
  \draw[->,color=blue!75,thick] (11,3) to [out=-100,in=90] node[] {} (11.8,1.8) ;
\end{tikzpicture}
\end{minipage}
}

One important feature of cascades is that  the enzyme on the first connected component of a certain layer is the last product of the first component of the previous layer. For instance, $S_{1,L_1}$ is the enzyme on the second layer and so on.
The corresponding reaction network for the $N$-layer cascade is the following

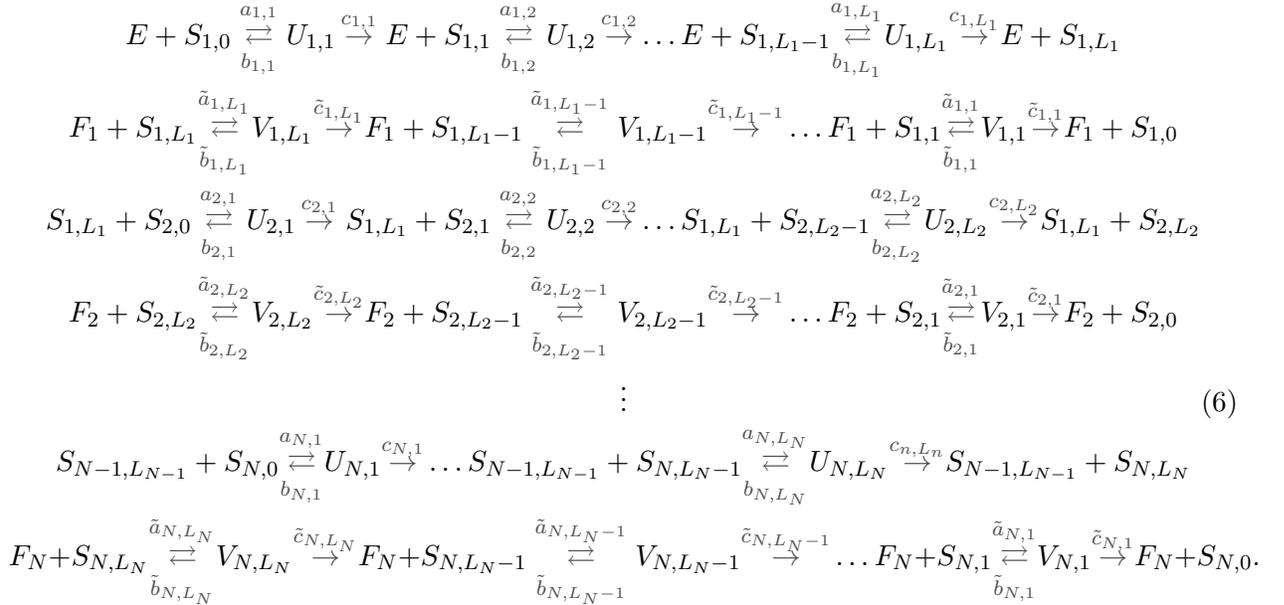
\begin{minipage}{.99\textwidth}
\begin{equation*}
E+S_{1,0} \; \textcolor{black!70}{\overset{a_{1,1}}{\underset{b_{1,1}}{\rightleftarrows}}} \; U_{1,1}  \, \textcolor{black!70}{\overset{c_{1,1}}{\rightarrow}} \;
E+S_{1,1}\; \textcolor{black!70}{\overset{a_{1,2}}{\underset{b_{1,2}}{\rightleftarrows}}} \; U_{1,2} \, \textcolor{black!70}{\overset{c_{1,2}}{\rightarrow}} \dots
E+S_{1,L_1-1}\textcolor{black!70}{\overset{a_{1,L_1}}{\underset{b_{1,L_1}}{\rightleftarrows}}} U_{1,L_1}
  \textcolor{black!70}{\overset{c_{1,L_1}}{\rightarrow}}E+S_{1,L_1}
\end{equation*}
\begin{equation*}
F_1+S_{1,L_1}\textcolor{black!70}{\overset{\tilde{a}_{1,L_1}}{\underset{\tilde{b}_{1,L_1}}{\rightleftarrows}}} V_{1,L_1} \textcolor{black!70}{\overset{\tilde{c}_{1,L_1}}{\rightarrow}}  F_1+ S_{1, L_1-1}
\; \textcolor{black!70}{\overset{\tilde{a}_{1,L_1-1}}{\underset{\tilde{b}_{1,L_1-1}}{\rightleftarrows}}} \; V_{1,L_1-1} \textcolor{black!70}{\overset{\tilde{c}_{1,L_1-1}}{\rightarrow}}\dots
F_1+S_{1,1}\textcolor{black!70}{\overset{\tilde{a}_{1,1}}{\underset{\tilde{b}_{1,1}}{\rightleftarrows}}} V_{1,1}
  \textcolor{black!70}{\overset{\tilde{c}_{1,1}}{\rightarrow}} F_1+S_{1,0}
\end{equation*}
\begin{equation*}
S_{1,L_1}+S_{2,0}\; \textcolor{black!70}{\overset{a_{2,1}}{\underset{b_{2,1}}{\rightleftarrows}}} \; U_{2,1} \, \textcolor{black!70}{\overset{c_{2,1}}{\rightarrow}} \;
  S_{1,L_1}+S_{2,1} \; \textcolor{black!70}{\overset{a_{2,2}}{\underset{b_{2,2}}{\rightleftarrows}}} \; U_{2,2} \, \textcolor{black!70}{\overset{c_{2,2}}{\rightarrow}} \dots
  S_{1,L_1}+S_{2,L_2-1}\textcolor{black!70}{\overset{a_{2,L_2}}{\underset{b_{2,L_2}}{\rightleftarrows}}} U_{2,L_2}
  \textcolor{black!70}{\overset{c_{2,L_2}}{\rightarrow}} S_{1,L_1}+S_{2,L_2}
\end{equation*}
\begin{equation*}
 F_2+S_{2,L_2}\textcolor{black!70}{\overset{\tilde{a}_{2,L_2}}{\underset{\tilde{b}_{2,L_2}}{\rightleftarrows}}} V_{2,L_2} \textcolor{black!70}{\overset{\tilde{c}_{2,L_2}}{\rightarrow}}  F_2+ S_{2, L_2-1}
\; \textcolor{black!70}{\overset{\tilde{a}_{2,L_2-1}}{\underset{\tilde{b}_{2,L_2-1}}{\rightleftarrows}}} \; V_{2,L_2-1} \textcolor{black!70}{\overset{\tilde{c}_{2,L_2-1}}{\rightarrow}} \dots
  F_2+S_{2,1}\textcolor{black!70}{\overset{\tilde{a}_{2,1}}{\underset{\tilde{b}_{2,1}}{\rightleftarrows}}} V_{2,1}
  \textcolor{black!70}{\overset{\tilde{c}_{2,1}}{\rightarrow}}F_2+S_{2,0}
\end{equation*}
\begin{equation}\label{eq:cascade}
 \vdots
\end{equation}
\begin{equation*}
S_{N-1,L_{N-1}}+S_{N,0}\textcolor{black!70}{\overset{a_{N,1}}{\underset{b_{N,1}}{\rightleftarrows}}} U_{N,1} \textcolor{black!70}{\overset{c_{N,1}}{\rightarrow}} 
\dots
S_{N-1,L_{N-1}}+S_{N,L_N-1}\textcolor{black!70}{\overset{a_{N,L_N}}{\underset{b_{N,L_N}}{\rightleftarrows}}} U_{N,L_N}\textcolor{black!70}{\overset{c_{n,L_n}}{\rightarrow}} S_{N-1,L_{N-1}}+S_{N,L_N}
\end{equation*}
\begin{equation*}
F_N+S_{N,L_N}\textcolor{black!70}{\overset{\tilde{a}_{N,L_N}}{\underset{\tilde{b}_{N,L_N}}{\rightleftarrows}}} V_{N,L_N} \textcolor{black!70}{\overset{\tilde{c}_{N,L_N}}{\rightarrow}} F_N+ S_{N, L_N-1}
\; \textcolor{black!70}{\overset{\tilde{a}_{N,L_N-1}}{\underset{\tilde{b}_{N,L_N-1}}{\rightleftarrows}}} \; V_{N,L_N-1} \textcolor{black!70}{\overset{\tilde{c}_{N,L_N-1}}{\rightarrow}}\dots
  F_N+S_{N,1}\textcolor{black!70}{\overset{\tilde{a}_{N,1}}{\underset{\tilde{b}_{N,1}}{\rightleftarrows}}} V_{N,1}
  \textcolor{black!70}{\overset{\tilde{c}_{N,1}}{\rightarrow}} F_N+S_{N,0}.
\end{equation*}

\end{minipage}

\bigskip

We will assume $F_i\neq F_j$ if $i\neq j$, and consider the following partition of the non-intermediate species, which satisfies \hypo~\ref{hyp:partition}:
\[
\Sp=\Sp^{(1)} \bigsqcup \Sp^{(2)} \bigsqcup \dots \bigsqcup \Sp^{(2N+1)},
\]
with $\Sp^{(m)}=\{S_{m,0}, \dots,S_{m,L_m} \}$ and $\Sp^{(N+m)}=\{F_m\}$, for $1\le m \le N$, and $\Sp^{(2N+1)}=\{E\}$.

\bigskip

As our running example for this section, we will consider the $2$-layer cascade with $18$ reactions.

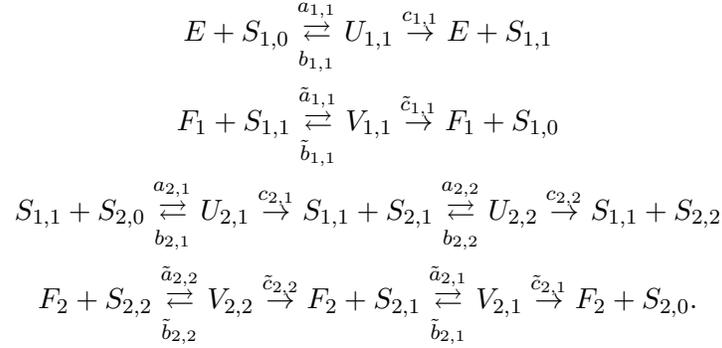
\begin{example}\label{ex:cascade_18reactions}
 \begin{equation*}
  E+S_{1,0} \overset{a_{1,1}}{\underset{b_{1,1}}{\rightleftarrows}}
  U_{1,1} \overset{c_{1,1}}{\rightarrow}
  E+S_{1,1}
 \end{equation*}
 \begin{equation*}
  F_1+S_{1,1} \overset{\tilde{a}_{1,1}}{\underset{\tilde{b}_{1,1}}{\rightleftarrows}}
  V_{1,1} \overset{\tilde{c}_{1,1}}{\rightarrow}
  F_1+S_{1,0}
 \end{equation*}
 \begin{equation*}
  S_{1,1}+S_{2,0} \overset{a_{2,1}}{\underset{b_{2,1}}{\rightleftarrows}}
  U_{2,1} \overset{c_{2,1}}{\rightarrow}
  S_{1,1}+S_{2,1} \overset{a_{2,2}}{\underset{b_{2,2}}{\rightleftarrows}}
  U_{2,2} \overset{c_{2,2}}{\rightarrow}
  S_{1,1}+S_{2,2}
 \end{equation*}
 \begin{equation*}
  F_2+S_{2,2} \overset{\tilde{a}_{2,2}}{\underset{\tilde{b}_{2,2}}{\rightleftarrows}}
  V_{2,2} \overset{\tilde{c}_{2,2}}{\rightarrow}
  F_2+S_{2,1} \overset{\tilde{a}_{2,1}}{\underset{\tilde{b}_{2,1}}{\rightleftarrows}}
  V_{2,1} \overset{\tilde{c}_{2,1}}{\rightarrow}
  F_2+S_{2,0}.
 \end{equation*}

 The first layer consists of two connected components. The first component consists of one modification performed by the enzyme $E$  on the substrate $S_{1,0}$, which is transformed into the product $S_{1,1}$. On the second connected component, the enzyme $F_1$ performs  the reverse modification on the substrate $S_{1,1}$. The second layer is similar.

For this network we have
 $\Sp^{(1)}=\{S_{1,0},S_{1,1}\}$, $\Sp^{(2)}=\{S_{2,0},S_{2,1},S_{2,2}\}$, $\Sp^{(3)}=\{F_1\}$, $\Sp^{(4)}=\{F_2\}$, and $\Sp^{(5)}=\{E\}$.
\end{example}

\subsection{Identifiability of constants in a general cascade}

The aim of this section is to show that all the constants in the cascades introduced in \eqref{eq:cascade} can be identified from successive derivatives of the variable corresponding to the last product of the last layer, $S_{N,L_N}$.
In order to prove this, we relate the derivatives of the last product of a given layer of the cascade with the derivatives of the last product of the layer immediately above.

To shorten notation, we will denote $K_{m, j} = b_{m,j}+ c_{m,j}$ and $\tilde{K}_{m, j} =  \tilde{b}_{m,j}+ \tilde{c}_{m,j}$ for every $1\le m\le N$, $1\le j \le L_m$. Also, for unifying purposes, we set $S_{0,L_{0}}:= E$.

For $1\le n\le N$, consider the variable $s_{n, L_n}$ corresponding to the last product of the $n$th layer of the cascade. We have that
\[
 \dot{s}_{n,L_n}=c_{n,L_n} u_{n,L_n} -\tilde{a}_{n,L_n} s_{n,L_n} f_n+\tilde{b}_{n,L_n}v_{n,L_n}- \sum_{j=1}^{L_{n+1}} a_{n+1,j} s_{n, L_n} s_{n+1, j-1}+ \sum_{j=1}^{L_{n+1}} K_{n+1,j} u_{n+1,j}
\]
and, for $n=N$, only the three first terms appear in the derivative, \textit{i.e.} $a_{N+1,j}=0$, $K_{N+1,j}=0$ for all $j$. The second derivative of $s_{n,L_n}$ is
\[
\begin{array}{c}
\ddot{s}_{n,L_n} = c_{n, L_n} (a_{n, L_n} \colorbox{red!30}{$s_{n-1, L_{n-1}}$}\, s_{n, L_n-1} - K_{n, L_n} u_{n, L_n})
   - \tilde{a}_{n, L_n} (\dot{s}_{n, L_n} f_n + s_{n, L_n} \dot{f}_n) + \qquad {} \\[2mm]
   \qquad {} + \tilde{b}_{n,L_n} (\tilde{a}_{n, L_n} s_{n, L_n} f_n  -\tilde{K}_{n, L_n}v_{n, L_n}) -  \displaystyle\sum_{j=1}^{L_{n+1}} a_{n+1, j} (\dot{s}_{n, L_n} s_{n+1, j-1} +s_{n, L_n} \dot{s}_{n+1, j-1}) + \\
 {}+ \displaystyle\sum_{j=1}^{L_{n+1}} K_{n+1, j} (a_{n+1, j} s_{n, L_n} s_{n+1, j-1} - K_{n+1, j} u_{n+1, j}).
\end{array}
\]

We can see that the variable $s_{n-1, L_{n-1}}$ corresponding to the last product of the $(n-1)$th layer appears in the second derivative of $s_{n, L_n}$. More precisely, from the above expression, it follows easily that it only appears in the term $c_{n, L_n} a_{n, L_n} s_{n-1, L_{n-1}} s_{n, L_n-1}$, since $S_{n-1, L_{n-1}}$ does not react with or to $F_n$ or $S_{n+1, j}$ for any $j$.
Thus, two differentiation steps enable us to ``jump'' from one layer of the cascade to the layer immediately above. Inductively, the idea is that, for $m<n$, by taking $2(n-m)$ derivatives of $s_{n, L_n}$ we will reach the $m$th layer; that is, the variable $s_{m, L_m}$ will appear and so, the successive derivatives of $s_{m, L_m}$ will appear in higher order derivatives of $s_{n, L_n}$.

Now, by the results in Section \ref{subsec:2comp} for the case of two connected components of the form \eqref{eq:2conncomp}, we can identify all constants in the $m$th layer of the cascade by looking at the coefficients of certain monomials of the derivatives of $s_{m, L_m}$.
Then, our previous considerations will imply that those constants can be identified from successive derivatives of $s_{n, L_n}$ as well. In order to ensure that this can be achieved, we prove that certain monomials effectively appear in the derivatives of $s_{n, L_n}$, and compute their coefficients (see Proposition \ref{propApp:monomials} in Appendix A for a precise statement and its proof).

When considering the last product of the last layer of the cascade, we obtain our main result:

\begin{theorem}\label{thm:cascade}
All the constants in the network~\eqref{eq:cascade}
can be identified from $s^{(\ell)}_{N,L_N}$ with $1\leq \ell \leq \max\{2N; 2(N-m+L_m)-1, \ 1\leq m\leq N\}$.
\end{theorem}

We now summarize the identifiability procedure which proves the previous theorem. The procedure obtains recursively, for $m= N, N-1, \dots, 1$,  the values of the constants $a_{m, j}$, $\tilde{a}_{m, j}$, $b_{m, j}$, $\tilde{b}_{m, j}$, $c_{m, j}$, and $\tilde{c}_{m, j}$, for $1\le j \le L_m$, from the successive derivatives of $s_{N, L_N}$, according to Table~\ref{table:cascade}.

In order to shorten notation, let $\cP_N:=1$, $\cC_N:= 1$, $\cK_N:=0$ and, for $1\le m \le  N-1$, $\cP_m:=\prod\limits_{i=m+1}^N s_{i, L_i-1}$, $\cC_m :=\prod\limits_{i=m+1}^N c_{i, L_i}a_{i, L_i}$ and $\cK_m := \sum\limits_{i=m+1}^N K_{i, L_i}$.

\begin{center}
\begin{table}[ht]

\renewcommand{\arraystretch}{1.8}

\begin{tabular}{|c|c|c|c|}
\hline
Derivative                 & Monomial          & Coefficient (up to sign)   & Constant       \\ \hline
\multirow{3}{*}
{${s}_{N,L_N}^{(2(N-m)+1)}$}
              & $u_{m,L_m}\cP_m$               & $c_{m,L_m} \cC_m$          & $c_{m,L_m}$         \\  \cline{2-4}
              & $f_m s_{m,L_m}\cP_m$           & $\tilde{a}_{m, L_m} \cC_m$ & $\tilde{a}_{m, L_m}$ \\  \cline{2-4}
              & $v_{m, L_m}\cP_m$              & $\tilde{b}_{m, L_m} \cC_m$ & $\tilde{b}_{m, L_m}$ \\  \hline
\multirow{5}{*}
{${s}_{N,L_N}^{(2(N-m)+2)}$}
    & $s_{m-1,L_{m-1}}s_{m,L_m-1}\cP_m$ & $c_{m,L_m} a_{m,L_m}\cC_m$        & $a_{m,L_m}$ \\ \cline{2-4}
    & $u_{m,L_m}\cP_m$                  & $c_{m,L_m}(K_{m,L_m}+\cK_m)\cC_m$ & $b_{m,L_m}$  \\  \cline{2-4}
    & $v_{m,L_m}\cP_m$  & $\tilde{b}_{m,L_m}(\tilde{K}_{m,L_m}+\cK_m)\cC_m$ & $\tilde{c}_{m,L_m}$ \\ \cline{2-4}
    & $f_m s_{m,j}s_{m,L_m}\cP_m$ & $\tilde{a}_{m,j}\tilde{a}_{m, L_m}\cC_m$ & $\tilde{a}_{m,j}$ {\scriptsize $(1\le j \le L_m-1)$} \\  \cline{2-4}
    & $v_{m,j}s_{m,L_m}\cP_m$   & $\tilde{K}_{m,j}\tilde{a}_{m,L_m}\cC_m$   & $\tilde{K}_{m,j}$ {\scriptsize $(1\le j \le L_m-1)$} \\  \hline
\multirow{2}{*}
{${s}_{N,L_N}^{(2(N-m)+3)}$}
    & $s_{m-1,L_{m-1}}s_{m,j-1}s_{m,L_m-1}\cP_m$ & $c_{m,L_m}a_{m,L_m}a_{m,j}\cC_m$ & $a_{m,j}$ {\scriptsize $(1\le j \le L_m-1)$}\\ \cline{2-4}
    & $u_{m,j} s_{m,L_m-1}\cP_m$                 & $c_{m,L_m}a_{m,L_m}K_{m,j}\cC_m$ & $K_{m,j}$ {\scriptsize $(1\le j \le L_m-1)$}\\ \hline
{${s}_{N,L_N}^{(2(N-m)+2k+1)}$}
    & $s_{m-1,L_{m-1}}^k u_{m,L_m-k}\cP_m$ & $c_{m,L_m-k} \prod\limits_{j=L_m-k+1}^{L_m}c_{m,j}a_{m,j}\, \cC_m$ & $c_{m,L_m-k}$, $b_{m,L_m-k}$ \\ \cline{2-4}
{\scriptsize $1\le k\le L_m-1$}
    & $s_{m-1,L_{m-1}}^k v_{m,L_m-k}\cP_m$& $\tilde{b}_{m,L_m-k} \prod\limits_{j=L_m-k+1}^{L_m}c_{m,j}a_{m,j}\, \cC_m$ & $\tilde{b}_{m,L_m-k}$, $\tilde{c}_{m,L_m-k}$ \\ \hline
\end{tabular}
\caption{The constants in the cascade can be identified from $s_{N,L_N}$. Tha table shows the monomials to be considered (column 2) in each of the successive derivatives of $s_{N,L_N}$ (column 1). For each monomial, taking into account the constants already identified, the corresponding coefficient (column 3) enables us to identify the constant appearing in the last column}\label{table:cascade}
\end{table}
\end{center}

\begin{example}[Example~\ref{ex:cascade_18reactions} continued]
Here we find the monomials relevant for identifiability in the 2-layer cascade. We highlight with blue boxes the constants that we are identifying in each derivative. We moreover highlight with green boxes the monomials that we used to identify $b_{1,1}$ and $\tilde{c}_{1,1}$ from $K_{1,1}$ and $\tilde{K}_{1,1}$, respectively (see rows 5 and 6 in Table~\ref{table:cascade}).

\begin{tabular}{rcl}
 $\dot{s}_{2,2}$ & $=$ & $-\fboxrule=1.2pt\fcolorbox{blue}{white}{$\tilde{a}_{2,2}$}f_2s_{2,2}+\fboxrule=1.2pt\fcolorbox{blue}{white}{$\tilde{b}_{2,2}$}v_{2,2}+\fboxrule=1.2pt\fcolorbox{blue}{white}{$c_{2,2}$}u_{2,2}$\\
 $\ddot{s}_{2,2}$ & $=$ & $-\tilde{a}_{2,2}[(-\tilde{a}_{2,2}f_2s_{2,2}+\tilde{K}_{2,2}v_{2,2}-\fboxrule=1.2pt\fcolorbox{blue}{white}{$\tilde{a}_{2,1}$}f_2s_{2,1}+\fboxrule=1.2pt\fcolorbox{blue}{white}{$\tilde{K}_{2,1}$}v_{2,1})s_{2,2}+f_2\dot{s}_{2,2}]+$\\
  & & $+\tilde{b}_{2,2}(\tilde{a}_{2,2}f_2s_{2,2}-\fboxrule=1.2pt\fcolorbox{blue}{white}{$\tilde{K}_{2,2}$}v_{2,2})+c_{2,2} (\fboxrule=1.2pt\fcolorbox{blue}{white}{$a_{2,2}$}s_{1,1}s_{2,1}-\fboxrule=1.2pt\fcolorbox{blue}{white}{$K_{2,2}$}u_{2,2})$ \\
 $s_{2,2}^{(3)}$ & $=$ & $\displaystyle\sum_{h+i \le 2} \beta_{f_2,h,i}\, f_2^{(h)} s_{2,2}^{(i)}+\delta_{v_{2,2}} v_{2, 2} +c_{2,2}(a_{2,2}[(\fboxrule=1.2pt\fcolorbox{blue}{white}{$c_{1,1}$}u_{1,1}-\fboxrule=1.2pt\fcolorbox{blue}{white}{$\tilde{a}_{1,1}$}f_1s_{1,1}+$\\
 & & $+\fboxrule=1.2pt\fcolorbox{blue}{white}{$\tilde{b}_{1,1}$}v_{1,1}-\fboxrule=1.2pt\fcolorbox{blue}{white}{$a_{2,1}$}s_{1,1}s_{2,0}+\fboxrule=1.2pt\fcolorbox{blue}{white}{$K_{2,1}$}u_{2,1}-a_{2,2}s_{1,1}s_{2,1}+K_{2,2}u_{2,2})s_{2,1}+$\\
 & & $+s_{1,1}(\fboxrule=1.2pt\fcolorbox{blue}{white}{$c_{2,1}$}u_{2,1}-a_{2,2}s_{1,1}s_{2,1}+b_{2,2}u_{2,2}+\tilde{c}_{2,2}v_{2,2}-\tilde{a}_{2,1}f_2s_{2,1}+\fboxrule=1.2pt\fcolorbox{blue}{white}{$\tilde{b}_{2,1}$}v_{2,1})]-$\\
 & & $-K_{2,2}(a_{2,2}s_{1,1}s_{2,1}-K_{2,2}u_{2,2}))$\\
 $s_{2,2}^{(4)}$ & $=$ & $\displaystyle\sum_{h+i \le 3} \beta_{f_2,h,i}\, f_2^{(h)} s_{2,2}^{(i)}+\delta_{v_{2,2}} v_{2, 2}+c_{2,2}(a_{2,2}[(c_{1,1}(\fboxrule=1.2pt\fcolorbox{blue}{white}{$a_{1,1}$}es_{1,0}-\fboxrule=1.2pt\fcolorbox{blue}{white}{$K_{1,1}$}\fboxrule=1.2pt\fcolorbox{green}{white}{$u_{1,1}$})-$\\
 & & $-\tilde{a}_{1,1}(\dot{f}_1s_{1,1}+f_1\dot{s}_{1,1})+\tilde{b}_{1,1}(\tilde{a}_{1,1}f_1s_{1,1}-\fboxrule=1.2pt\fcolorbox{blue}{white}{$\tilde{K}_{1,1}$}\fboxrule=1.2pt\fcolorbox{green}{white}{$v_{1,1}$})-a_{2,1}(\dot{s}_{1,1}s_{2,0}+s_{1,1}\dot{s}_{2,0})+$\\
 & & $+K_{2,1}\dot{u}_{2,1}-a_{2,2}(\dot{s}_{1,1}s_{2,1}+s_{1,1}\dot{s}_{2,1})+K_{2,2}\dot{u}_{2,2})s_{2,1}+2\dot{s}_{1,1}\dot{s}_{2,1}+s_{1,1}\ddot{s}_{2,1}]-$\\
 & & $ -K_{2,2}(a_{2,2}[(c_{1,1}\fboxrule=1.2pt\fcolorbox{green}{white}{$u_{1,1}$}-\tilde{a}_{1,1}f_1s_{1,1}+\tilde{b}_{1,1}\fboxrule=1.2pt\fcolorbox{green}{white}{$v_{1,1}$}-a_{2,1}s_{1,1}s_{2,0}+K_{2,1}u_{2,1}-$\\
 & & $-a_{2,2}s_{1,1}s_{2,1}+K_{2,2}u_{2,2})s_{2,1}+s_{1,1}\dot{s}_{2,1}]-K_{2,2}\dot{u}_{2,2})$ \\
\end{tabular}

\end{example}

\subsection{An example of how to obtain the rate constants from data}\label{sec:algorithm}

Here, we will illustrate our previous theoretical identifiability results in a specific example, showing how they can be used as a guidance in experimental design for practical parameter identification from observable data.

The 3-layer cascade with $L_1=1, L_2=L_3=2$ represents the well known MAPK signaling cascade with $s_{3,2}$ representing the concentration of the doubly phosphorylated kinase MAPK-PP \cite{CDVS16,sig-016,kholo00,sig-051}. Consider, then, the cascade~\eqref{eq:cascade} for $N=3$ and $L_1=1$, $L_2=L_3=2$, whose schematic representation is introduced at the beginning of Section~\ref{sec:cascade}. In this case, we have $22$ species  concentrations $\bx$ and 30 rate constants $\Ka$ which can be identified from $s^{(\ell)}_{3,2}$, $1\leq \ell \leq 6$, by Theorem~\ref{thm:cascade}. According to Definition~\ref{identvar}, this means that if we consider the polynomial system
\begin{equation}\label{eq:lhs}
 s_{3,2}^{(\ell)}(\bx,\Ka)=p_{\ell}(\bx,\Ka)
\end{equation}
for the corresponding polynomials $p_{\ell}$ obtained from~\eqref{CRN} by computing the successive total derivatives of $s_{3,2}$, the function that maps the vector of rate constants $\Ka$ to the coefficients of the polynomials $p_{\ell}$'s (considered as polynomials in the species concentrations $\bx$) is injective. This means that all the rate constants can be recovered from noise-free data by a suitable interpolation procedure: if we evaluate these polynomials at ``sufficiently many'' points $\bx\in\R^{22}$, we may reconstruct the coefficients and, consequently, determine uniquely the values of the rate constants. 

However, it is not clear which $\bx\in\R^{22}$ are suitable for identifying the parameters of the system, nor how many of them are enough for this purpose. We give here a heuristic to choose a list of $\bx\in\R^{22}$ based on the monomials in the second column of Table~\ref{table:30constants}, which is the adapted version of Table~\ref{table:cascade} for this particular case. This heuristic can be used as an aid to design experiments to obtain the rate constants values. Each initial state $\bx\in\R^{22}$ is in correspondence with a different experiment.

\medskip

In order to recover the value of the $30$ rate constants in this case, we propose the following algorithm:
\begin{enumerate}
 \item[Step~1.] Consider $\bx_1,\bx_2,\dots,\bx_{30}\in\R^{22}$ defined as follows: for the $i$-th monomial in Table~\ref{table:30constants}, consider $\bx_i\in\R^{22}$ where all the coordinates are $0$ except for those coordinates corresponding to variables that divide the monomial, which are equal to $1$. For example, for the monomial $u_{1,1}s_{2,1}s_{3,1}$, all the coordinates of the associated point are equal to $0$, except for the three coordinates corresponding to $u_{1,1}$, $s_{2,1}$ and $s_{3,1}$ that are equal to $1$.
 \item[Step~2.] For each $i\in\{1,\dots,30\}$, obtain the value $s_{3,2}^{(\ell)}(\bx_i,\Ka)$ for the order $\ell$ that corresponds to the $i$-th monomial in Table~\ref{table:30constants}. Ideally, these values should be obtained experimentally, for instance considering $\bx_i$ the initial state at time $t=0$.
 \item[Step~3.] Construct a (nonlinear) polynomial equation system from~\eqref{eq:lhs}, of $30$ equations in the $30$ unknowns $\Ka$, by evaluating the right-hand sides at $\bx_1,\dots,\bx_{30}$ and replacing the left-hand sides with the values obtained in the previous step.
 \item[Step~4.] Solve the polynomial system in the unknowns $\Ka$.
\end{enumerate}

A vague explanation of why this heuristic works is that each monomial in Table~\ref{table:30constants}  incorporates a new variable that comes paired with the new rate constant to be identified. Further research is needed to find a rigorous proof for this conjecture.

We implemented the algorithm above by reconstructing the values of the left-hand sides of~\eqref{eq:lhs} with the rate constants in the third column of Table~S2 in the Supporting Information of \cite{qiao07}. We used Maple~\cite{maple} to solve the system of equations and successfully obtained the following values (in a few seconds using a standard desktop computer).

\medskip

\begin{center}
 \begin{minipage}{0.95\textwidth}
  \noindent $a_{1,1} = 337.2299998$, $a_{2,1} = 1226.000001$, $a_{2,2} = 3383.7$, $a_{3,1} = 229.5699981$, $a_{3,2} = 3388.7$, $\tilde{a}_{1,1} = 1841.000002$, $\tilde{a}_{2,1} = 2960.300016$, $\tilde{a}_{2,2} = 1956.8$, $\tilde{a}_{3,1} = 297.0$, $\tilde{a}_{3,2} = 974.7$, $b_{1,1} = 261.1000013$, $b_{2,1} = 623.1700002$, $b_{2,2} = 605.3100002$, $b_{3,1} = 694.13$, $b_{3,2} = 485.3499999$, $\tilde{b}_{1,1} = 198.47$, $\tilde{b}_{2,1} = 163.0$, $\tilde{b}_{2,2} = 48.804$, $\tilde{b}_{3,1} = 301.09$, $\tilde{b}_{3,2} = 587.45$, $c_{1,1} = 146.07$, $c_{2,1} = 420.0000001$, $c_{2,2} = 214.65$, $c_{3,1} = 43.658$, $c_{3,2} = 65.732$, $\tilde{c}_{1,1} = 338.4400021$, $\tilde{c}_{2,1} = 668.2000111$, $\tilde{c}_{2,2} = 67.97000003$, $\tilde{c}_{3,1} = 31.743$, $\tilde{c}_{3,2} = 175.91$.
 \end{minipage}
\end{center}

\medskip

The same $3$-layer cascade may be completely identified also by means of the result stated in Theorem~\ref{thm:eachconncomp}: in this case the rate constants in each connected component can be identified from $s^{(\ell_1)}_{1,1}$,  $s^{(\ell_2)}_{1,0}$, $s^{(\ell_3)}_{2,2}$, $s^{(\ell_4)}_{2,0}$, $s^{(\ell_5)}_{3,2}$, and $s^{(\ell_6)}_{3,0}$ respectively, for $1\leq \ell_1,\ell_2\leq 2$ and $1\leq \ell_3,\ell_4,\ell_5,\ell_6\leq 3$. By Corollary~\ref{coro:2comp} we can also identify the constants from $s^{(\ell_1)}_{1,1}$, $s^{(\ell_3)}_{2,2}$, and $s^{(\ell_5)}_{3,2}$, for $1\leq \ell_1\leq 2$ and $1\leq \ell_3,\ell_5\leq 3$. We adapted the procedure above and implemented it in Maple, and we obtained the same rate constants as before.

Throughout the article we assume that one can use noise-free data in order to recover the rate constants values. Nevertheless, there are certain numerical errors that appear at Step~4, when the polynomial system in the unknowns $\Ka$ is solved. If we moreover implement the algorithm with numerical approximations of the total derivatives, more numerical errors are bound to occur. The major drawback of considering the last two approaches, based on Theorem~\ref{thm:eachconncomp} or Corollary~\ref{coro:2comp}, is that more species have to be measured. However, the value that has to be numerically estimated corresponds to a derivative of order at most three, which can be approximated more accurately and with fewer time measurements than those values of higher order derivatives.

\medskip

The Maple code for both procedures can be found at \newline
 http://cms.dm.uba.ar/Members/mpmillan/identifiability.

\medskip

\begin{center}
\begin{table}

\renewcommand{\arraystretch}{1.5}
\begin{tabular}{|c|c|c|c|}
\hline
Derivative   & Monomial          & Coefficient (up to sign)  & Constant       \\ \hline
\multirow{3}{*}
{$\dot{s}_{3,2}$} & $u_{3,2}$       & $c_{3,2}$              & $c_{3,2}$         \\  \cline{2-4}
                  & $f_3s_{3,2}$    & $\tilde{a}_{3,2}$      & $\tilde{a}_{3,2}$ \\  \cline{2-4}
                  & $v_{3,2}$       & $\tilde{b}_{3,2}$      & $\tilde{b}_{3,2}$ \\  \hline
\multirow{5}{*}
{$\ddot{s}_{3,2}$} & $s_{2,2} s_{3,1}$   & $c_{3,2} a_{3,2}$                 & $a_{3,2}$         \\  \cline{2-4}
                   & $u_{3,2}$           & $c_{3,2}K_{3,2}$                  & $b_{3,2}$ \\  \cline{2-4}
                   & $v_{3,2}$           & $\tilde{b}_{3,2}\tilde{K}_{3,2}$  & $\tilde{c}_{3,2}$ \\  \cline{2-4}
                   & $f_3s_{3,1}s_{3,2}$ & $\tilde{a}_{3,1}\tilde{a}_{3,2}$  & $\tilde{a}_{3,1}$ \\  \cline{2-4}
                   & $v_{3,1}s_{3,2}$    & $\tilde{K}_{3,1}\tilde{a}_{3,2}$  & $\tilde{K}_{3,1}$ \\  \hline
\multirow{8}{*}
{$s_{3,2}^{(3)}$} & $s_{2,2} s_{3,0}s_{3,1}$ & $c_{3,2} a_{3,2} a_{3,1}$         & $a_{3,1}$\\ \cline{2-4}
          & $u_{3,1} s_{3,1}$        & $c_{3,2} a_{3,2} K_{3,1}$         & $K_{3,1}$\\ \cline{2-4}
                  & $s_{2,2} u_{3,1}$        & $c_{3,1} c_{3,2} a_{3,2}$         & $c_{3,1}, b_{3,1}$\\  \cline{2-4}
                  & $v_{3,1} s_{2,2}$        & $\tilde{b}_{3,1} c_{3,2} a_{3,2}$ & $\tilde{b}_{3,1},\tilde{c}_{3,1}$ \\ \cline{2-4}
                  & $u_{2,2} s_{3,1}$        & $c_{2,2} c_{3,2} a_{3,2}$         & $c_{2,2}$ \\  \cline{2-4}
                  & $f_2 s_{2,2}s_{3,1}$     & $\tilde{a}_{2,2} c_{3,2} a_{3,2}$ & $\tilde{a}_{2,2}$ \\  \cline{2-4}
                  & $v_{2,2} s_{3,1}$        & $\tilde{b}_{2,2} c_{3,2} a_{3,2}$ & $\tilde{b}_{2,2}$ \\ \hline
\multirow{5}{*}
{$s_{3,2}^{(4)}$} & $s_{1,1} s_{2,1} s_{3,1}$  & $c_{2,2}a_{2,2} c_{3,2} a_{3,2}$    & $a_{2,2}$ \\ \cline{2-4}
          & $u_{2,2} s_{3,1}$          & $c_{2,2} (K_{2,2}+K_{3,2}) \cC_2$ & $b_{2,2}$ \\ \cline{2-4}
                  & $v_{2,2} s_{3,1}$          & $\tilde{b}_{2,2} (\tilde{K}_{2,2}+K_{3,2}) \cC_2$  & $\tilde{c}_{2,2}$ \\ \cline{2-4}
                  & $f_2s_{2,1}s_{2,2}s_{3,1}$ & $\tilde{a}_{2,1}\tilde{a}_{2,2}\cC_2$ & $\tilde{a}_{2,1}$\\ \cline{2-4}
                  & $v_{2,1}s_{2,2}s_{3,1}$    & $\tilde{K}_{2,1}\tilde{a}_{2,2}\cC_2$ & $\tilde{K}_{2,1}$\\ \hline
\multirow{7}{*}
{$s_{3,2}^{(5)}$} & $u_{1,1} s_{2,1} s_{3,1}$      & $c_{1,1}\cC_1$  & $c_{1,1}$ \\ \cline{2-4}
          & $f_1s_{1,1}s_{2,1}s_{3,1}$     & $\tilde{a}_{1,1}\cC_1$ & $\tilde{a}_{1,1}$\\ \cline{2-4}
          & $v_{1,1}s_{2,1}s_{3,1}$        & $\tilde{b}_{1,1}\cC_1$ & $\tilde{b}_{1,1}$\\ \cline{2-4}

          & $s_{1,1}s_{2,0}s_{2,1}s_{3,1}$ & $c_{2,2}a_{2,2}a_{2,1}\cC_2$ & $a_{2,1}$\\ \cline{2-4}
          & $u_{2,1}s_{2,1}s_{3,1}$        & $c_{2,2}a_{2,2}K_{2,1}\cC_2$ & $K_{2,1}$ \\ \cline{2-4}
          & $s_{1,1}u_{2,1}s_{3,1}$        & $c_{2,1}c_{2,2}a_{2,2}a_{3,2}c_{3,2}$ & $c_{2,1}$, $b_{2,1}$\\ \cline{2-4}
          & $s_{1,1}v_{2,1}s_{3,1}$        & $\tilde{b}_{2,1}\cC_1$ & $\tilde{b}_{2,1}$,$\tilde{c}_{2,1}$\\ \hline
\multirow{3}{*}
{$s_{3,2}^{(6)}$} & $es_{1,0} s_{2,1} s_{3,1}$ & $c_{1,1}a_{1,1}\cC_1$   & $a_{1,1}$ \\ \cline{2-4}
                  & $u_{1,1} s_{2,1} s_{3,1}$  & $c_{1,1}(K_{1,1}+K_{2,2}+K_{3,2})\cC_1$  & $b_{1,1}$ \\ \cline{2-4}
                  & $v_{1,1}s_{2,1}s_{3,1}$    & $\tilde{b}_{1,1}(\tilde{K}_{1,1}+K_{2,2}+K_{3,2})\cC_1$ & $\tilde{c}_{1,1}$ \\ \hline
\end{tabular}

\caption{The constants in the the 3-layer cascade with 30 constants can be identified from
$s_{3,2}$. Tha table shows the monomials to be considered (column 2) in each of the successive derivatives of $s_{3,2}$ (column 1). For each monomial, taking into account the constants already identified, the corresponding coefficient (column 3) enables us to identify the constant appearing in the last column.  Here, we consider $\cC_1= c_{2,2}a_{2,2}c_{3,2}a_{3,2}$ and $\cC_2=c_{3,2}a_{3,2}$.
}\label{table:30constants}
\end{table}
\end{center}

\section{Discussion and further work}

The main contribution of this paper has been to prove that all the rate constants in several well known chemical reaction networks, that are abundant in the literature, can be identified from a reduced set of kinetic variables.
The work here extends previous results by Craciun and Pantea \cite{CP} and avoids computationally expensive procedures such as differential elimination and Gr\"obner basis \cite{BSAD07,Bo07,MEDS09}.

We should point out that we assumed that there is a special partition of the set of chemical species, and that every connected component of the chemical reaction network has a particular shape (see Section~\ref{sec:assumptions}). Both assumptions are natural when modeling multisite phosphorylation systems and signaling cascades \cite{ws08,sig-016}. We have then shown, in Section~\ref{sec:components}, how to identify the rate constants in every connected component, or two related connected components, from a single species. In Section~\ref{sec:cascade} we have moreover proved that all the rate constants in signaling cascade networks can be identified from only one species: the last product of the first component of the last layer. Additionally, we have presented in Section~\ref{sec:algorithm} an example showing how to compute the values of  the rate constants from noise free data according to our theoretical results in the previous sections. The procedure is based on a heuristic to choose the right input data; it would be of great interest to find a formal proof for establishing a good set of sufficient data for any network of the class considered in this paper.

We expect that the techniques used in this paper could be applied for identifiability from a few variables to a number of modifications of the networks we have considered here. For instance, it would be interesting to introduce more intermediate complexes within different reactions. Another potential adaptation is relaxing the assumption $F_i\neq F_j$ for $i\neq j$ in the cascade network, and allowing for repetition of these enzymes. Both modifications are natural extensions of the networks we have analized and we conjecture that similar results can be obtained. We moreover would like to apply our techniques to more general but hence well structured networks such as MESSI networks \cite{PMD18}. Another future research direction is to characterize which other variables can be considered to identify the rate constants of either a whole connected component or the entire biochemical network.

\bigskip

\noindent \textbf{Acknowledgements}
 The authors wish to thank the anonymous referees for their thoughtful comments which helped to improve the manuscript. This is a pre-print of an article published in the Bulletin of Mathematical Biology. The final authenticated version is available online at: https://doi.org/10.1007/s11538-019-00594-0.

\newpage

\appendix

\section{Proofs}\label{sec:pf}

Throughout this Appendix, we maintain the notation and assumptions introduced in Sections \ref{sec:preliminaries} and \ref{sec:assumptions}.

Before stating and proving our results, we introduce some further notation and formulas we will use in our analysis. We consider an autonomous dynamical system
\begin{equation}
\label{eqApp:CRN}
{\dot \bx}= \underset{y\to y'}{\sum} k_{yy'} \,  \bx^y \, (y'-y),
\end{equation}
arising from a chemical reaction network satisfying the assumptions stated in Section \ref{sec:assumptions}.

For a non-intermediate species $X$, let
\begin{equation}\label{eq:reactants}
\Zp_X = \{Z : Z \mbox{ reacts with } X\} \quad \hbox{and} \quad \Wp_X = \{ W : W \mbox{ reacts to } X\}.
\end{equation}
By the shape of the networks we consider, $\Zp_X$ is a set of non-intermediate species and $\Wp_X$ is a set of intermediate species.
From \eqref{eqApp:CRN}, we then have that
\begin{equation}\label{eq:dotx}
\dot{x}=-\sum_{Z\in \Zp_X} \mu_{z} x z+ \sum_{W\in \Wp_X}\eta_{w}w,
\end{equation}
for suitable non negative real numbers $\mu_z$ and $\eta_w$.
For $\ell \ge 2$, Leibniz rule implies that
\begin{equation}\label{eq:deriv_xl}
 x^{(\ell)}=-\sum_{Z\in\Zp_X}\mu_{z}\sum_{h+i=\ell-1}\binom{\ell-1}{h}x^{(h)}z^{(i)}+
 \sum_{W\in \Wp_X}\eta_{w}w^{(\ell-1)}.
\end{equation}
If $W\in \Wp_X$ is involved in a block of reactions
 \[Z_{w,1}+Z_{w,2} \overset{a_w}{\underset{b_w}{\rightleftarrows}} W \overset{c_w}{\rightarrow} Z_{w,3}+X,\]
then, according to \eqref{eqApp:CRN}, the differential equation $\dot{w}=a_wz_{w,1}z_{w,2}-K_w w,$ with $K_w = b_w+c_w$, is satisfied, and
\begin{equation}\label{eq:deriv_ul}
 w^{(\ell-1)}= \sum_{h+i\le \ell -2} (-K_w)^{\ell-2-h-i}a_w\binom{h+i}{h}z_{w,1}^{(h)}z_{w,2}^{(i)}+(-K_w)^{\ell-1} w.
\end{equation}
By separating the cases where $X\in \{Z_{w,1}, Z_{w,2}\}$ and $X\notin\{Z_{w,1}, Z_{w,2}\}$, we can simplify:
\begin{equation}\label{eq:simple_l}
 x^{(\ell)}=\sum_{Z\in\Zp_X} \sum_{h+i\le \ell - 1} \beta_{z,h,i}\,x^{(h)}z^{(i)}
+\sum_{W\in \Wp_X \atop X\notin\{Z_{w,1},Z_{w,2}\}}
\sum_{h+i\le \ell-2}\gamma_{w,h,i}\,z_{w,1}^{(h)}z_{w,2}^{(i)} + \sum_{w\in \Wp_X} \delta_w \, w,
\end{equation}
for suitable real numbers $\beta_{z,h,i}, \gamma_{w,h,i}$ and $\delta_w$ that depend on $\ell$ and the reaction rate constants.

From the previous formulas interpreted as polynomials in the variables $x, z, w$, we deduce straightforwardly:

\begin{lemma}\label{lem:cte} For a reaction network satisfying the assumptions of Section \ref{sec:assumptions}, we have:
\begin{enumerate}
\item The constant monomial does not appear in any derivative of any species.
\item The only monomials of degree $1$ appearing in a derivative $x^{(\ell)}$, $\ell \ge 1$, for a non-intermediate species $X$, are the monomials $w$ corresponding to $W\in \Wp_X$, that is, those that appear in $\dot x$.
\end{enumerate}
\end{lemma}

\subsection{Proofs of Section \ref{subsec:1comp}: Identifying the constants in one connected component from one variable}

Here we give the proofs of our identifiability result for a connected component of the type:
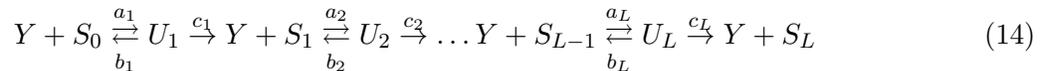
\begin{equation}\label{eqApp:conncomp}
Y+S_0\overset{a_1}{\underset{b_1}{\rightleftarrows}} U_1 \overset{c_1}{\rightarrow}
Y+S_1\overset{a_2}{\underset{b_2}{\rightleftarrows}} U_2 \overset{c_2}{\rightarrow} \dots
Y+S_{L-1}\overset{a_L}{\underset{b_L}{\rightleftarrows}} U_L
\overset{c_{L}}{\rightarrow}Y+S_L
\end{equation}
We maintain the hypotheses and notations introduced in Section \ref{sec:assumptions} and previously in this Appendix.

\begin{lemma}\label{lem:ajKj}
Given a connected component as in \eqref{eqApp:conncomp},
the constants $a_L, b_L$ and $c_L$ can be identified from $\dot{s}_L$ and $\ddot{s}_L$, and, if $L>1$, the constants $a_j$ and $K_j:=b_j+c_j$, for $1\le j \le L-1$, can be identified from $\dot{s}_L, \ddot{s}_L$ and $s_L^{(3)}$.
\end{lemma}

\begin{proof}
Following (\ref{eq:dotx}), we have
\begin{equation}\label{eq:dot_xn+1}
\dot{s}_{L}=-\sum_{Z\in \Zp_L} \mu_{z}s_{L}z+ \sum_{W\in \Wp_L}\eta_{w}w
\end{equation}
where, using the notation in (\ref{eq:reactants}), $\Zp_L:= \Zp_{S_L}$ and $\Wp_L:= \Wp_{S_L}$.
By separating the term corresponding to $U_L\in \Wp_L$, we obtain
\[\dot{s}_{L}=-\sum_{Z\in \Zp_L} \mu_{z}s_{L}z+ c_L u_L+\sum_{W\in \Wp_L^*}\eta_{w}w
 \]
where $\Wp_L^*:= \Wp_L\setminus \{U_L\}$. Then, we can  identify  $c_L$ from $\dot{s}_L$ as the coefficient of the monomial $u_L$.

Consider now
\begin{equation*}
 \ddot{s}_{L}=-\sum_{Z\in \Zp_L} \mu_{z} [\dot{s}_{L}z+s_{L}\dot{z}]+
 c_L(a_Lys_{L-1}-K_Lu_L)+\sum_{W\in \Wp_L^*}\eta_{w} [a_w z_{w,1} z_{w,2}-K_w w].
\end{equation*}
From this expression, since $c_L \ne 0$, we can identify $a_L$ and $K_L$ from the coefficients of the monomials $ys_{L-1}$ and $u_L$ (which only appear in $\ddot{s}_L$ from the derivative $\dot{u}_L$) and, as we know $c_L$, we can also identify $b_L$. If $L=1$ we have identified all the constants.

If $L>1$, consider the third derivative
\begin{equation}\label{eq:d3_xn+1}
\displaystyle
\begin{array}{rcl}
s^{(3)}_{L}&=&\displaystyle\sum_{Z\in \Zp_L} \sum_{h+i\le 2} \beta_{z,h,i} \, s_{L}^{(h)} z^{(i)} +
 c_L a_L[\dot{y}s_{L-1}+y\dot{s}_{L-1}]-c_L K_L(a_Lys_{L-1}-K_Lu_L)+ \\
 & & {} + \displaystyle\sum_{W\in \Wp_L^* \atop S_L \notin \{ Z_{w,1}, Z_{w,2}\}}\gamma_{w,1} [\dot{z}_{w,1} z_{w,2}+ z_{w,1} \dot{z}_{w,2}] + \gamma_{w,0} z_{w,1} z_{w,2} + \delta_w w.
\end{array}
\end{equation}

The constants $a_j$ and $K_j$, for $1\le j \le L-1$, appear in
$\dot y =  \sum_{1\le j \le L} (- a_j y s_{j-1} + K_j u_j) + \cdots$
as the coefficients (up to sign) of the monomials $y s_{j-1}$ and $u_j$ respectively.
Then, they appear in the expression (\ref{eq:d3_xn+1}) from the product $\dot{y} s_{L-1}$ in the coefficients of the monomials  $y s_{j-1} s_{L-1}$ and $u_{j}s_{L-1}$, for $1\leq j\leq L-1$.
We will now look for these monomials in the whole expression (\ref{eq:d3_xn+1}) and show that they come only from the product $\dot{y} s_{L-1}$.

As $Y \notin \Zp_L$ and, for every $Z\in \Zp_L$, by \hypo~\ref{hyp:partition},
we have $Z\neq S_l$ for all $0\leq l\leq L-1$, the monomials $y s_{j-1} s_{L-1}$ and $u_{j}s_{L-1}$, for $1\leq j\leq L-1$, do not appear in $s_L^{(h)} z$, for $0\le h \le 2$. Also, it is clear that they do not appear in $s_L z^{(i)}$, for $0\le i \le 2$. On the other hand, every monomial of degree $3$ that appears in a product of two derivatives of order $1$ is a multiple of an intermediate; so, $y s_{j-1} s_{L-1}$ does not appear in $\dot{s}_L \dot{z}$ and, by Lemma \ref{lem:cte}, $u_j s_{L-1}$ does not appear either since no derivative contains a constant term or the degree one monomial $s_{L-1}$.

Now, consider $W\in \Wp_L^*$ such that $S_L\notin \{Z_{w,1}, Z_{w,2}\}$, and the corresponding block of reactions $Z_{w,1}+Z_{w,2}\rightleftarrows W \rightarrow Z_{w,3} + S_L$. Since $U_j\notin \Wp_L^*$ for every $1\le j \le L$, then $Z_{w,1}+Z_{w,2} \ne Y +S_{j-1}$. Also, by \hypo~\ref{hyp:partition},
$Z_{w,1}+Z_{w,2} \ne S_{l} +S_{L-1}$ for every $0\le l\le L$.
Every monomial in $\dot{z}_{w,1} z_{w,2}$ is either of the form $w_0 z_{w,2}$ for an intermediate $W_0$ that reacts to $Z_{w,1}$ or of the form $z_0 z_{w,1} z_{w,2}$ for a non-intermediate $Z_0$ that reacts with $Z_{w,1}$. If $z_0 z_{w,1} z_{w,2} = y s_{j-1} s_{L-1}$, it follows that $Z_{w,1} +Z_{w,2}\in \{ Y+S_{j-1}, Y+S_{L-1}, S_{j-1}+S_{L-1}\}$, leading to a contradiction.
If $w_0 z_{w,2} = u_j s_{L-1}$, then $Z_{w,2} = S_{L-1}$ and $U_j$ reacts to $Z_{w,1}$, meaning that $Z_{w,1}\in \{Y, S_{j-1}, S_j\}$, which is not possible.

Finally, the monomial $ys_{j-1} s_{L-1}$ does not appear in $y \dot{s}_{L-1}$ since, by \hypo~\ref{hyp:partition},
$S_{j-1}$ does not react with $S_{L-1}$ for every $j$.

We conclude that, for $1\le j \le L-1$, the coefficients in $s_L^{(3)}$ of the monomials $ys_{j-1} s_{L-1}$ and $u_j s_{L-1}$ are $ -c_L a_L a_j$ and $c_L a_L K_j$, respectively.
As we have already identified $c_L$ and $a_L$, these coefficients enable us to identify $a_j$ and $K_j=b_j+c_j$, for $1\le j \le L-1$.
\end{proof}

\bigskip

We show now some auxiliary results concerning the behavior of monomials appearing in the successive derivatives of some variables and their relations with the reaction network. They will allow us to prove Lemma \ref{lem:xn_several} below,
the key recursive tool to show the identifiability results of Section \ref{subsec:1comp}.

\begin{lemma}\label{lem:prod_no_interm}
If $\underset{i=1}{\overset{m}{\prod}}z_i$, with $m\ge 2$, is a monomial of $x^{(\ell)}$  where $Z_i$ is a non-intermediate species for every $i$, then there exist $1\leq i_1<i_2\leq m$ such that $Z_{i_1}$ reacts with $Z_{i_2}$.
\end{lemma}

\begin{proof}
 If $\ell=1$, this is true. Assume $\ell \ge 2$. Recalling that $x^{(\ell)} = \sum\limits_{v} \frac{\partial{x^{(\ell-1)}}}{\partial v} \dot v$ (where the sum runs over all variables $v$ representing non-intermediates or intermediate species),  it follows that $\underset{i=1}{\overset{m}{\prod}}z_i$ is a monomial in $\frac{\partial{x^{(\ell-1)}}}{\partial v} \dot v$ for some variable $v$. Since every monomial appearing in $\dot v$ is either a single intermediate or a product of two non-intermediate species that react together, the result follows.
\end{proof}

\begin{corollary}\label{cor:powers}
 If $X$ is a non-intermediate species, no derivative $x^{(\ell)}$ for $\ell \ge 1$ contains a monomial which is a pure power of degree $m\ge 2$ of a variable corresponding to a non-intermediate species.
\end{corollary}

\begin{lemma} \label{lemma:structure}
Given an intermediate species $U$ and non-intermediate species  $X$ and $Y$ such that $Y\ne X$, if a monomial $y^r u$, $r\ge 0$, appears in $x^{(\ell)}$ for some $\ell \ge 1$, then  either $U$ reacts to $X$ or $\ell \ge 2$, the network contains a block of reactions
\begin{equation}\label{eq:block}
Y+Z_w \rightleftarrows W \rightarrow  \widetilde Z_w + X,
\end{equation}
where $Z_w\ne X$, and a monomial $y^{t} u$ with $t<r$ appears in $z_w^{(i)}$ for some $i\le \ell-2$.
If, in addition, $Y$ acts as an enzyme in all the reactions of the connected component determined by $U$, then $X \in \Sp_U$ and the block of reactions in \eqref{eq:block} is $Y+Z_w \rightleftarrows W \rightarrow  Y + X$, and it is contained in the connected component determined by $U$.

Moreover, if $U$ does not react to $X$ and $\ell$ is the smallest integer such that a monomial $y^r u$ appears in $x^{(\ell)}$, then $r\ge 1$, $\ell \ge 2$,  and the monomial $y^{r-1} u $ appears in $z_w^{(i)}$ for some $i\le \ell-2$.
\end{lemma}

\begin{proof} We prove the first part by induction on $r$. If $r=0$, then $u$ appears in $x^{(\ell)}$ for some $\ell \ge 1$;  by Lemma \ref{lem:cte} (2), this is equivalent to the fact that $U$ reacts to $X$. In particular, if $Y\ne X$ acts as an enzyme in the connected component determined by $U$, then $X\in \Sp_U$.

Now, if $r\ge 1$, since no monomial $y^r u$ with $r\ge 1$ appears in $\dot x$, it follows that $\ell \ge 2$. Then, by identity (\ref{eq:simple_l}), the monomial $y^ru$ can only appear in a product of derivatives of two species, and by Lemma \ref{lem:cte} and Corollary \ref{cor:powers}, one of these species must be $Y$ and the corresponding order of derivation must be zero.

If $y^r u$ appears in a product $x^{(h)} z^{(i)}$ for some $Z\in \Zp_X$ and $h+i \le \ell -1$, as $X\ne Y$, then $Z=Y$ and $y^{r-1} u$ appears in $x^{(h)}$; then, the result follows by the inductive hypothesis.

Finally, if $y^r u$ appears in a product $z_{w,1}^{(h)} z_{w,2}^{(i)}$ with $h+i \le \ell -2$ for some $W \in \Wp_X$ such that $X\notin\{ Z_{w,1}, Z_{w, 2}\}$, again by Lemma \ref{lem:cte} and Corollary \ref{cor:powers}, we may assume that $Y= Z_{w,1}$ and $y^{r-1} u$ appears in $z_{w,2}^{(i)}$. Since $X\ne Z_{w,2}$ and $W$ reacts to $X$, we must have
 $Y+Z_{w,2} \rightleftarrows W \rightarrow \widetilde Z_w + X$ for some species $\widetilde Z_w$, that is, a block of reactions as in \eqref{eq:block}. By the induction hypothesis applied to the non-intermediate $Z_{w,2} \ne Y$, if $Y$ acts as an enzyme in the connected component determined by $U$, it follows that $Z_{w,2}\in \Sp_U$. Then, $Y+Z_{w,2}$ is a complex in the connected component determined by $U$, where $Y$ acts as an enzyme. As $X\ne Y$, necessarily $\widetilde Z_{w} =Y$ and $X\in \Sp_U$.

To see that the last statement of the lemma holds, note that if $U$ does not react to $X$ and a monomial $y^r u$ appears in a derivative $x^{(\ell)}$, then $r\ge 1$ and $\ell \ge 2$ and, by assuming $\ell$  minimal, the only possibility in the above reasoning is the last one.
\end{proof}

Now, we are able to prove the key lemma for the proof of our main result on the identifiability of constants in a single connected component. We keep our previous notation and assumptions.

For technical reasons, we define the empty product of factors $\alpha_i$ as $\underset{i=0}{\overset{-1}{\prod}}\alpha_i=1$.

\begin{lemma}\label{lem:xn_several}
Given a connected component as in \eqref{eqApp:conncomp},
with $L\ge 1$, let $1\leq n\leq L$ and $0\leq k \leq n-1$ be fixed.
If $\ell$ is minimum such that $y^r u_{n-k}$ is a monomial of $s^{(\ell)}_n$ for some $r \ge 0$, then
$\ell=2k+1$, $r=k$ and the coefficient of $y^k u_{n-k}$ in $s^{(2k+1)}_n$ is
\[c_{n-k} \prod_{j=0}^{k-1} a_{n-j} \, c_{n-j}.\]
\end{lemma}

\begin{proof} For $k= 0$,  first notice that, for all $1\leq n \leq L$, as $U_n$ reacts to $S_n$, then $u_n$ appears in $\dot{s}_n$ and so, $\ell=1$, $r=0$, and the coefficient of $u_{n}$ is $c_n$, as we wanted to prove.

We follow the proof by induction on $n$.

If $n=1$, the only possibility is $k=0$, which we have already proven.

Assume now $n\ge 2$, and let $k\geq 1$.  If a monomial $y^r u_{n-k}$ appears in $s_n^{(\ell)}$ and considering $\ell$ minimal,  as $U_{n-k}$ does not react to $S_n$, by Lemma \ref{lemma:structure} 
applied to $U:=U_{n-k}$ and $X:= S_n$,
the network contains a block of reactions $Y+Z_w \rightleftarrows W \rightarrow  Y + S_n,$
and the monomial $y^{r-1} u_{n-k}$ appears in $z_w^{(i)}$ for some $i\le \ell-2$. This block of reactions is necessarily $Y+S_{n-1} \rightleftarrows U_n \rightarrow   Y+ S_n$ and so, $y^{r-1} u_{n-k}$ appears in $s_{n-1}^{(i)}$ for some $i\le \ell-2$. Moreover, by formula (\ref{eq:simple_l}) applied to $x=s_{n}$, the only terms contributing to the monomial $y^r u_{n-k}$ come from products $y s_{n-1}^{(i)}$ with $i\le \ell-2$. Since $y^{r-1} u_{n-k} = y^{r-1} u_{(n-1)-(k-1)}$, by the induction hypothesis,  $i\ge 2(k-1)+1= 2k-1$; then, $\ell -2\ge 2k-1$ or, equivalently, $\ell \ge 2k +1$.

Consider now formula (\ref{eq:simple_l}) for $s_n^{(2k+1)}$. The only product of derivatives where a monomial  $y^r u_{n-k}$ may appear is $y s_{n-1}^{(2k-1)}$, since $i\le 2k-1$ for all derivatives $s_{n-1}^{(i)}$ involved. Then, the coefficient of $y^r u_{n-k}$ in $s_n^{(2k+1)}$ equals $\gamma_{u_n, 0, 2k-1}$ multiplied by the coefficient of $y^{r-1} u_{n-k}$ in $s_{n-1}^{(2k-1)}$.
By the induction hypothesis, a monomial $y^{r-1}u_{n-k}$ appears with nonzero coefficient in $s_{n-1}^{(2k-1)}$ if and only if $r-1 = k-1$, that is $r=k$, and the corresponding coefficient is $c_{n-1-(k-1)}\prod_{j=0}^{k-2} a_{n-1-j} c_{n-1-j}$.
To determine $\gamma_{u_n, 0, 2k-1}$ note that, by formula (\ref{eq:deriv_ul}) applied to $u_n$, the product $y s_{n-1}^{(2k-1)}$ appears in $u_n^{(2k)}$ multiplied by $a_n$ and, by formula  (\ref{eq:deriv_xl}), $u_n^{(2k)}$ appears in $s_n^{(2k+1)}$ multiplied by $c_n$; then,
 $\gamma_{u_n, 0, 2k-1} = c_n a_n$.

Summarizing, the monomial $y^k u_{n-k}$ appears with  nonzero coefficient in $s_n^{(2k+1)}$; hence, $\ell = 2k+1$. Moreover, it is the only monomial of the form $y^r u_{n-k}$ effectively appearing  in $s_n^{(2k+1)}$, and its corresponding coefficient is
$c_n a_n c_{n-k}\prod_{j=0}^{k-2} a_{n-1-j}\, c_{n-1-j} = c_{n-k} \prod_{j=0}^{k-1} a_{n-j} \, c_{n-j}.$
\end{proof}

\begin{remark}
An interesting fact is that the previous lemmas also hold for networks where not all the reactions are enzymatic. By this we mean that the blocks of reactions are of the form:
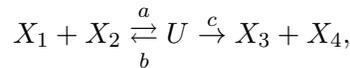
\[X_1+X_2\overset{a}{\underset{b}{\rightleftarrows}}U \overset{c}{\to} X_3+X_4,\]
with $X_1\neq X_2$, $X_3\neq X_4$ but not necessarily $\{X_1,X_2\}\cap\{X_3,X_4\}\neq \emptyset$.
\end{remark}

Combining Lemmas \ref{lem:ajKj} and \ref{lem:xn_several}, we may now prove the main result of Section  \ref{subsec:1comp} (Proposition \ref{prop:1conncomp} in the main text):

\begin{proposition} \label{propApp:1conncomp}
All the constants in a connected component as \eqref{eqApp:conncomp} of a network satisfying the assumptions in Section \ref{sec:assumptions} can be identified from $s^{(\ell)}_{L}$ with $1\leq \ell \leq \mathrm{max}\{2,2L-1\}$.
\end{proposition}

\begin{proof}
By Lemma \ref{lem:ajKj}, we can identify $a_L, b_L$ and $c_L$ from $\dot{s}_L$ and $\ddot{s}_{L}$, which implies the statement of the proposition for $L=1$.

For $L\ge 2$, again by Lemma \ref{lem:ajKj}, we can also identify $a_j$ and $ K_j = b_j +c_j$, for $1\le j \le L-1$, from $s_L^{(3)}$. In order to identify all the constants, we need to ``separate'' $b_j$ from $c_j$ for $1\leq j\leq L-1$.
We do this by identifying the constants $c_{L-k}$ recursively, for $k=1,\dots, L-1$, from the successive derivatives of $s_L$.

Let $k\ge 1$ and assume $c_{L-j}$ has been identified, for $0\le j<k$.  By Lemma~\ref{lem:xn_several},
the coefficient of the monomial $y^k u_{L-k}$ in $s^{(2k+1)}_L$ is
$c_{L-k} \prod_{j=0}^{k-1} a_{L-j} \, c_{L-j}.$
As $ a_{L-j}$ and $c_{L-j}$ for $0\le j \le L-1$ are known, from this coefficient we identify $c_{L-k}$.
\end{proof}

\subsection{Proofs of Section \ref{subsec:2comp}: Identifying the constants in two connected components from one variable}

The case of two connected components of the type
\begin{equation}\label{eqApp:2conncomp}
 \begin{array}{c}
  Y+S_0\overset{a_1}{\underset{b_1}{\rightleftarrows}} U_1 \overset{c_1}{\rightarrow}
  Y+S_1\overset{a_2}{\underset{b_2}{\rightleftarrows}} U_2 \overset{c_2}{\rightarrow} \dots
  Y+S_{L-1}\overset{a_L}{\underset{b_L}{\rightleftarrows}} U_L
  \overset{c_{L}}{\rightarrow}Y+S_L,\\
  \widetilde Y +S_L\overset{\tilde{a}_L}{\underset{\tilde{b}_L}{\rightleftarrows}} V_L \overset{\tilde{c}_L}{\rightarrow}
  \widetilde Y+S_{L-1}\overset{\tilde{a}_{L-1}}{\underset{\tilde{b}_{L-1}}{\rightleftarrows}} V_{L-1} \overset{\tilde{c}_{L-1}}{\rightarrow} \dots
  \widetilde Y+S_1\overset{\tilde{a}_1}{\underset{\tilde{b}_1}{\rightleftarrows}} V_1
  \overset{\tilde{c}_1}{\rightarrow}\widetilde Y+S_0
 \end{array}
 \end{equation}
considered in the paper runs in a similar way than the one connected component case.
The first result concerning this class of networks is in the spirit of Lemma \ref{lem:ajKj}.

\begin{lemma}\label{lem:tildeajKj}
Given two connected components as in \eqref{eqApp:2conncomp}, the constants $\tilde{a}_L$, $\tilde{b}_L$, $\tilde{c}_L$, and $\tilde{a}_j, \tilde{K}_j:=\tilde{b}_j + \tilde{c}_j$, for $1\le j\le L-1$, can be identified from $\dot{s}_L$ and $\ddot{s}_L$.
\end{lemma}

\begin{proof} Consider the formula for $\dot{s}_L$ given in \eqref{eq:dot_xn+1}. Separating the terms corresponding to $\widetilde Y\in \Zp_L$ and $V_L\in \Wp_L$, and writing $\Zp_L^{\times}:= \Zp_L\setminus \{ \widetilde Y\}$ and $\Wp_L^{\times}:= \Wp_L\setminus \{ V_L\}$, we obtain
\[
 \dot{s}_L=-\sum_{Z\in \Zp_L^{\times}} \mu_{z}s_L z+ \sum_{W\in \Wp_L^{\times}}\eta_{w}w
 -\tilde{a}_L \tilde{y} s_L + \tilde{b}_L v_L.
\]
Then, we can identify $\tilde{a}_L$ and $\tilde{b}_L$ as the coefficients (up to sign) of the monomials $\tilde{y} s_L$ and $v_L$ in $\dot{s}_L$.

Consider now
\begin{equation}\label{eq:ddot_sn+1}
 \ddot{s}_L=-\sum_{Z\in \Zp_L^\times} \mu_{z} [\dot{s}_L z+s_L \dot{z}]+
 \sum_{W\in \Wp_L^\times}\eta_{w} [a_w z_{w,1} z_{w,2}-K_w w] - \tilde{a}_L [\dot{\tilde{y}} s_L+\tilde{y} \dot{s}_L]+ \tilde{b}_L[\tilde{a}_L \tilde{y} s_L -\tilde{K}_L v_L].
 \end{equation}
From the coefficient of $v_L$ in $\ddot{s}_L$, we can identify $\tilde{K}_L$ and, therefore, $\tilde{c}_L$, since we have already identified $\tilde{b}_L$.
The constants $\tilde{a}_j$ and $\tilde{K}_j$, for $1\le j \le L-1$, appear in the derivative
\[\dot{\tilde{y}} =  \sum_{1\le j \le L} (- \tilde{a}_j \tilde{y} s_{j} + \tilde{K}_j v_j) + \cdots , \]
then, they appear in the expression \eqref{eq:ddot_sn+1} from the product $\dot{\tilde{y}} s_L$ in the coefficients of the monomials $\tilde{y} s_j s_L$ and $v_js_L$, respectively. By \hypo~\ref{hyp:partition},
$S_j \notin \Zp_L$ for every $1\le j\le L-1$; hence, the monomials $\tilde{y} s_j s_L$  do not come from any other term in \eqref{eq:ddot_sn+1}. Also, it is immediate that the monomials $v_j s_L$ only come from $\dot{\tilde{y}} s_L$. Then, the coefficients of $\tilde{y} s_j s_L$ and $v_j s_L$ in $\ddot{s}_L$ are $\tilde{a}_L \tilde{a}_j$ and $ -\tilde{a}_L \tilde{K}_j$, respectively, and enable us to identify $\tilde{a}_j$ and $\tilde{K}_j$, for $1\le j \le L-1$, since $\tilde{a}_L \ne 0$.
\end{proof}

In order to establish a statement extending Lemma \ref{lem:xn_several} to this new setting,  we need a previous technical lemma (a suitable analogue of Lemma \ref{lemma:structure} above):

\begin{lemma}\label{lemma:conn_comp_segundo}
 Given an intermediate species $V$ and a non-intermediate species
 $Y$ that acts as an enzyme in a connected component where the set
 of substrates and products is $\Sp_V$, if
  $X$ is a non-intermediate species such that
 $X\in\Sp^{(\alpha)}$ for some $\alpha\geq 1$ and $Y,\widetilde Y \notin\Sp^{(\alpha)}$, where $\widetilde Y$ is the enzyme in the connected component determined by $V$, and
 the monomial $y^r v$ appears in $x^{(\ell)}$ for some $r\ge 0$ and $\ell\geq1$,  then $X\in\Sp_V$.

Moreover, either $V$ reacts to $X$ or $r\ge 1$, $\ell \ge 2$ and a monomial $y^{t} v $ with $t<r$ appears in $z_w^{(i)}$, for some $i\le \ell - 2$, for a species $Z_w$ involved in a block of reactions $Y+Z_w \rightleftarrows W \rightarrow  Y + X.$ If $r\ge 1$ and $\ell$ is minimal, then $t=r-1$.
\end{lemma}

\begin{proof} First, note that $X\ne Y$ and $X\ne \widetilde Y$, because of the assumption that $X\in \Sp^{(\alpha)}$ and $Y, \widetilde Y \notin \Sp^{(\alpha)}$.
We proceed by induction on $r\in \N_0$.

If $r=0$, by Lemma \ref{lem:cte} (2),  $V$ reacts to $X$. As $X$ is not the enzyme $\widetilde Y$, then  $X\in\Sp_V$.

For $r\geq 1$, since $X\ne Y$, Lemma \ref{lemma:structure} states that either $V$ reacts to $X$ (which we have already considered) or the network contains a block of reactions $Y+Z_w \rightleftarrows W \rightarrow  \widetilde Z_w + X$, where $Z_w \ne X$, and a monomial $y^t v$ with $t<r$  appears in $z_w^{(i)}$ for some $i\le \ell-2$ (furthermore, $t= r-1$ if $\ell$ is minimal).
In the latter case, $\widetilde Z_w$ acts as an enzyme in the connected component determined by $W$ and $X\in \Sp_W$, which implies that $\Sp_W  \subset \Sp^{(\alpha)}$. If $\widetilde Z_w = Z_w$, then $Y \in \Sp_W $, contradicting the assumption that $Y \notin \Sp^{(\alpha)}$; therefore, $\widetilde Z_w = Y$, and $Z_w \in \Sp^{(\alpha)}$. By the induction hypothesis, $Z_w\in \Sp_V$. As $\Sp_V$ is the set of substrates and products in a connected component where $Y$ acts as an enzyme,
the complex $Y+Z_{w}$ lies in that component, and so, $X\in\Sp_V$.
\end{proof}

We are now able to prove the result that will play the key role in order to give a recursive argument to identify all the constants in suitable pairs of connected components:

\begin{lemma}\label{lemApp:xn_several_segundo} \label{lem:xn_several_segundo}
Given two connected components as in \eqref{eqApp:2conncomp}
with $L\ge 1$, let $1\leq n\leq L$ and $0\leq k \leq n-1$ be fixed.
If $\ell$ is minimum such that $y^r v_{n-k}$ is a monomial of $s^{(\ell)}_n$ for some $r \ge 0$, then $\ell=2k+1$, $r=k$ and the coefficient of $y^k v_{n-k}$ in $s_n^{(2k+1)}$ is
\[\tilde{b}_{n-k}\prod_{j=0}^{k-1} a_{n-j}c_{n-j} .\]
\end{lemma}

\begin{proof}
For $k=0$, and all $1\le n \le L$, $v_{n}$ appears in $\dot{s}_n$ (since $V_n$ reacts to $S_n$) with coefficient $\tilde{b}_{n}$ and so, $r=0$ and $\ell=1$.  We now proceed by induction on $n$.

If $n=1$, the only possibility is $k=0$, which has already been considered.

For $n\ge 2$, let $k\ge 1$. By \hypo~\ref{hyp:partition},
there exists $\alpha\ge 1$ such that $S_j \in \Sp^{(\alpha)}$ for every $0\le j \le L$, and $Y, \widetilde Y \notin \Sp^{(\alpha)}$. If the monomial $y^rv_{n-k}$ appears in a derivative of $s_n$ and $\ell$ is the minimum derivation order where it appears, as $V_{n-k}$ does not react to $S_n$, by Lemma \ref{lemma:conn_comp_segundo}, $r\ge 1$, $\ell\ge 2$ and the monomial $y^{r-1}v_{n-k}$ appears in $z_w^{(i)}$, for some $i\le \ell-2$, for a species $Z_w$ in a block of reactions $Y+Z_w \rightleftarrows W \rightarrow  Y + S_n$. Then, $W= U_n$ and $Z_w = S_{n-1}$;  so, $y^{r-1} v_{n-k}$ appears in $s_{n-1}^{(i)}$ for some $i\le \ell-2$. By the induction hypothesis, we have that $i\ge 2k-1$; therefore, $\ell\ge 2k+1$.

Now, following \emph{mutatis mutandis} the proof of Lemma \ref{lem:xn_several}, we deduce that the coefficient of the monomial $y^k v_{n-k}$ in $s_n^{(2k+1)}$ is equal to $c_n a_n$ multiplied by the coefficient of $y^{k-1} v_{n-k}$ in $s_{n-1}^{(2k-1)}$, and we conclude by applying the induction hypothesis.
\end{proof}

Similarly as in the previous subsection, from Lemmas \ref{lem:tildeajKj} and \ref{lem:xn_several_segundo} we deduce the following identifiability result for two connected components that extends Proposition \ref{propApp:1conncomp} and constitutes the main result in Section \ref{subsec:2comp} (Proposition \ref{prop:2comp} in the main text):

\begin{proposition}\label{propApp:2comp}
Given a chemical reaction network satisfying the assumptions in Section~\ref{sec:assumptions}, all the constants in two connected components as in \eqref{eqApp:2conncomp} can be identified from $s^{(\ell)}_{L}$ with $1\leq \ell \leq \mathrm{max}\{2,2L-1\}$.
\end{proposition}

\begin{proof} The result holds for $L=1$, since by Lemmas \ref{lem:ajKj} and \ref{lem:tildeajKj}, we  can identify $a_L, b_L, c_L, \tilde{a}_L, \tilde{b}_L$ and $\tilde{c}_L$ from $\dot{s}_L$ and $\ddot{s}_L$.

Assume now $L\ge 2$. By Proposition \ref{propApp:1conncomp}, all the constants $a_j, b_j$ and $c_j$, for $1\le j \le L$, can be identified from $s_L^{(\ell)}$ with $1\le \ell \le \max\{2, 2L-1\}$. It remains to show that we can also identify $\tilde{a}_j, \tilde{b}_j$ and $\tilde{c}_j$, for $1\le j \le L$.

By Lemma \ref{lem:tildeajKj}, the constants  $\tilde{a}_L$, $\tilde{b}_L$,  $\tilde{c}_L$ and $\tilde{a}_j$ and $\tilde{K}_j=\tilde{b}_j + \tilde{c}_j$,  for $1\le j \le L-1$,  are identifiable from $\dot{s}_L$ and $\ddot{s}_L$. We just need to ``separate'' $\tilde{b}_j$ and $\tilde{c}_j$ for $1\le j \le L-1$.  Due to Lemma
\ref{lem:xn_several_segundo}, this can be done by identifying $\tilde{b}_{L-k}$ recursively, for $k=1,\dots, L-1$, from the coefficients of the monomials $y^kv_{L-k}$ in $s_L^{(2k+1)}$.
\end{proof}

\subsection{Proofs of Section \ref{sec:cascade}: Identifying the cascade}

The following two auxiliary technical lemmas will be used in subsequent arguments concerning the identifiability in the cascade.

\begin{lemma}\label{lem:prod_no_interm_partes}
 If $\underset{j=1}{\overset{M}{\prod}}z_j$, with $Z_j$ non-intermediate species for all $j$, is a monomial of $x^{(\ell)}$ for a non-intermediate species $X\in\Sp^{(\alpha)}$ and $\ell \ge 1$,  then
there exists $1\leq j_1, j_2\leq M$ such that $Z_{j_1}\in\Sp^{(\alpha)}$ and $Z_{j_2} \in \Sp^{(\beta)}$ for some $\beta$ such that the network contains a complex $X+Z$ with $Z\in \Sp^{(\beta)}$.
\end{lemma}

\begin{proof}
For $\ell =1 $ the result is true, since the only products of non-intermediate species appearing in $\dot{x}$ are of the form $xz$ for a species $Z$ that reacts with $X$. Assume the lemma holds for derivatives of order $1\le h\leq \ell-1$ of non-intermediate species.

By equation~\eqref{eq:simple_l}, if the monomial appears in $x^{(h)}z^{(i)}$ for some $h+i \le  \ell-1$ and $h>0$, by Lemma \ref{lem:cte}(1), there is a monomial $\prod_{l=1}^{M'} z_{j_l}$ in $x^{(h)}$ with $1\le h\le \ell -1$, and the induction hypothesis gives the result. Assume now $h=0$ and $x= z_M$. If the monomial appears in $x z^{(i)}$ with $i\le \ell-1$, either $i=0$, and the monomial is $xz$ with $X$ and $Z$ reacting together, which implies the statement, or $1\le i \le \ell-1$ and the monomial  $\prod_{j=1}^{M-1} z_j$ appears in $z^{(i)}$. If the latter holds, $Z_{j_1} = X \in \Sp^{(\alpha)}$ and,  by the induction hypothesis applied to $Z \in \Sp^{(\beta)}$ for some $\beta$ and $1\le i \le \ell-2$, there exists $j_2$ such that $Z_{j_2} \in \Sp^{(\beta)}$.

If the product appears in $z_{w,1}^{(h)}z_{w,2}^{(i)}$, for some $h+ i\leq \ell-2$, coming from a block of reactions  $Z_{w,1}+Z_{w,2}\rightleftarrows W \to Z_{w,3}+X$ with $X\notin\{Z_{w,1},Z_{w,2}\}$, then the enzyme is $Z_{w,3}$ and, assuming $Z_{w,1}=Z_{w,3}$, it follows that $Z_{w,2}$ and $X$ lie in $\Sp_W$. Since $X\in \Sp^{(\alpha)}$, then $\Sp_W\subset \Sp^{(\alpha)}$; in particular, $Z_{w,2}\in \Sp^{(\alpha)}$. On the other hand, $Z_{w,1}= Z_{w,3} \in \Sp^{(\beta)}$ for some $\beta \ne \alpha$.  If $i=0$, there exists $1\le j_1\le M$ such that $Z_{j_1} = Z_{w,2} \in \Sp^{(\alpha)}$ and, if $i\ge 1$, by the induction hypothesis applied to $Z_{w,2}\in \Sp^{(\alpha)}$ and the factor of the monomial appearing in $z_{w,2}^{(i)}$, there exists $j_1$ such that $Z_{j_1}\in \Sp^{(\alpha)}$. Similarly, if $h=0$, there exists $1\le  j_2\le M$ such that $ Z_{j_2} = Z_{w,1} \in \Sp^{(\beta)}$ and, if $h\ge 1$, by the induction hypothesis applied to $Z_{w,1}\in \Sp^{(\beta)}$, there exists $j_2$ such that $Z_{j_2}\in \Sp^{(\beta)}$.
\end{proof}

\begin{lemma}\label{lem:prod_1_interm}
If $u\underset{i=1}{\overset{M}{\prod}}z_i$, with $M\ge 1$, is a monomial of $x^{(\ell)}$ for $U$ an intermediate species and $X, Z_i$ non-intermediate species for all $i$, either
there exist $1\leq i_1<i_2\leq M$ such that $Z_{i_1}$ reacts with $Z_{i_2}$ or there exist $1\le i_0\le M$ and a species $V$ that reacts with $Z_{i_0}$ such that $U$ reacts to a complex containing $V$.
\end{lemma}

\begin{proof}
Note that there are no monomials of this type in $\dot{x}$; thus, $\ell\ge 2$. For $\ell =2$, the only monomials in $\ddot{x}$ that are multiples of an intermediate and non-intermediates are:
\begin{itemize}
\item $uz$, for an intermediate species $U$ that reacts to $X$ and a non-intermediate $Z$ that reacts with $X$. In this case, the statement holds with $V=X$ and $Z_{i_0} = Z$;
\item $u x$, for an intermediate species $U$ that reacts to a non-intermediate species $Z$ reacting with $X$. The statement holds with $V=Z$ and $Z_{i_0}=X$.
\end{itemize}
For $\ell>2$, recalling that $x^{(\ell)} = \sum\limits_{v} \frac{\partial{x^{(\ell-1)}}}{\partial v} \dot v$ (where the sum runs over all variables $v$ representing non-intermediates or intermediate species),  it follows that $u \underset{i=1}{\overset{M}{\prod}}z_i$ is a monomial in $\frac{\partial{x^{(\ell-1)}}}{\partial v} \dot v$ for some variable $v$. Every monomial in $\dot v$ is either a single intermediate or a product of two non-intermediate species in a reaction. In the second case, the result follows.
Now, if $\underset{i=1}{\overset{M}{\prod}}z_i$ is a monomial of $\frac{\partial{x^{(\ell-1)}}}{\partial v}$ and $u$ is a monomial of $\dot v$, we have that $v\underset{i=1}{\overset{M}{\prod}}z_i$ is a monomial of $x^{(\ell-1)}$ and one of the following possibilities for $V$:
\begin{itemize}
\item $V= U$; then, $u\underset{i=1}{\overset{M}{\prod}}z_i$ is a monomial of $x^{(\ell-1)}$ and the result follows by induction.
\item $V$ is a non-intermediate species such that $U$ reacts to a complex containing $V$. By Lemma \ref{lem:prod_no_interm}, there are two variables in $v\underset{i=1}{\overset{M}{\prod}}z_i$ that react together. If none of these variables is $v$, there exist $1\le i_1, i_2\le M$ such that $Z_{i_1}$ and $Z_{i_2}$ react together; otherwise, there exists $1\le i_0\le M$ such that $V$ reacts with $Z_{i_0}$.
\end{itemize}
\end{proof}



\bigskip
We follow here the notations introduced in Section \ref{sec:cascade}, more precisely, in the general cascade (\ref{eq:cascade}). We also set $S_{0,L_0}:=E$.

For $1\le n\le N$, we have
\[
 \dot{s}_{n,L_n}=c_{n,L_n} u_{n,L_n} -\tilde{a}_{n,L_n} s_{n,L_n} f_n+\tilde{b}_{n,L_n}v_{n,L_n}- \sum_{j=1}^{L_{n+1}} a_{n+1,j} s_{n, L_n} s_{n+1, j-1}+ \sum_{j=1}^{L_{n+1}} K_{n+1,j} u_{n+1,j}
\]
and, for $n=N$, only the three first terms appear in the derivative, \textit{i.e.} $a_{N+1,j}=0$, $K_{N+1,j}=0$ for all $j$.

For $\ell \ge 2$, following equation \eqref{eq:simple_l}:
\begin{equation}\label{eq:dl_last_product}
\begin{array}{rcl}
s^{(\ell)}_{n,L_n} &=& \displaystyle\sum_{h+i \le \ell -1} \beta_{f_n,h,i}\, s_{n, L_n}^{(h)} f_n^{(i)} +
\displaystyle\sum_{j=1}^{L_{n+1}}\sum_{h+i \le \ell-1} \beta_{s_{n+1, j-1},h,i}\, s_{n, L_n}^{(h)} s_{n+1, j-1}^{(i)}+ \\
& + & \displaystyle\sum_{h+i\le \ell -2} \gamma_{u_{n, L_n},h,i}\, s_{n-1, L_{n-1}}^{(h)} s_{n, L_n-1}^{(i)}+\delta_{u_{n, L_n}} u_{n, L_n} + \delta_{v_{n,L_n}} v_{n, L_n} + \displaystyle\sum_{j=1}^{L_{n+1}} \delta_{u_{n+1,j}} u_{n+1, j}
\end{array}
\end{equation}
where
\[ \beta_{f_n,h,i} = \begin{cases} - \binom{\ell-1}{h} \tilde{a}_{n, L_n} & \mbox{ if } h+i= \ell -1\\ \tilde{a}_{n, L_n}\tilde{b}_{n, L_n} \binom{h+i}{h} (-\tilde{K}_{n, L_n})^{\ell-2-h-i} & \mbox{ if } h+i\le \ell -2\end{cases}, \]
\[ \beta_{s_{n+1, j-1},h,i} = \begin{cases} - \binom{\ell-1}{h} a_{n+1, j} & \mbox{ if } h+i= \ell -1\\ - \binom{h+i}{h} a_{n+1, j}  (-K_{n+1, j})^{\ell-1-h-i} & \mbox{ if } h+i\le \ell -2\end{cases}, \]
\[ \gamma_{u_{n,L_n},h,i} = c_{n,L_n} a_{n, L_n} \binom{h+i}{h} (-K_{n, L_n})^{\ell-2-h-i} \quad \hbox{for } 0\le h+i\le \ell-2,\]
\[\delta_{u_{n, L_n}} = c_{n, L_n} (-K_{n, L_n})^{\ell-1}, \quad \delta_{v_{n, L_n}} = \tilde{b}_{n, L_n} (-\tilde{K}_{n, L_n})^{\ell-1},\]
 \[ \delta_{u_{n+1, j}} = (-1)^{\ell-1} K_{n+1, j}^{\ell} \quad \hbox{for } 0 \le h+i\le \ell -2\]

\bigskip

According to formula \eqref{eq:dl_last_product}, every monomial of $s_{n, L_n}^{(\ell)}$ is either an intermediate species that appears in $\dot{s}_{n, L_n}$,
or it appears as a monomial in one of the products:
\begin{enumerate}
\item[(a)] $ s_{n, L_n}^{(h)} f_n^{(i)}$ for $h+i\le \ell-1$,
\item[(b)] $s_{n, L_n}^{(h)} s_{n+1, j-1}^{(i)} $ $(1\le j \le L_{n+1})$ for $h+i \le \ell -1$,
\item[(c)] $ s_{n-1, L_{n-1}}^{(h)} s_{n, L_n-1}^{(i)}$ for $h+i\le \ell-2$.
\end{enumerate}

The following three technical lemmas describe how the coefficients of some distinguished monomials change recursively after differentiation. These results allow us to obtain Proposition~\ref{propApp:monomials} below and hence, the identifiability result about the cascade stated in Section \ref{sec:cascade} (Theorem \ref{thm:cascade} in the main text).

\begin{lemma}\label{lem:recursion_sn}
Let $\cM = \prod_{j=1}^M z_j$ be a monomial of $s_{n-1, L_{n-1}}^{(\ell_0)}$ which is not a monomial of any derivative of $s_{n-1, L_{n-1}}$ of lower order and only involves variables corresponding to species in $\Sp^{(k)}$, $\Sp^{(N+k)}$, for $1\le k \le n-1$, and $\Sp^{(2N+1)}$. Assume  that :
\begin{itemize}
\item $\mathcal{M}$ is square-free and does not involve two disjoint pairs of variables corresponding to species that react together;
\item if $s_{n-1, L_{n-1}}$ divides $\mathcal{M}$, for every $1\le j_1, j_2\le M$ such that $Z_{j_1}$ and $Z_{j_2}$ react together, $Z_{j_1} = s_{n-1, L_{n-1}}$ or $Z_{j_2} = s_{n-1, L_{n-1}}$.
\end{itemize}
Then, $\widehat \cM := s_{n, L_n-1} \cM$ is a monomial of $s_{n, L_n}^{(\ell_0+2)}$ and of no lower order derivative of $s_{n, L_n}$. Moreover, if $C_\cM$ is the coefficient of $\cM$ in $s_{n-1, L_{n-1}}^{(\ell_0)}$, the coefficient of $\widehat \cM$ in  $s_{n, L_n}^{(\ell_0+2)}$ is $c_{n, L_n} a_{n, L_n} C_\cM$.
\end{lemma}

\begin{proof}
Assume $\widehat{\cM}$ is a monomial of $s_{n, L_n}^{(\ell)}$ for some $\ell \ge 1$. Then, it is a monomial of one of the products in cases (a), (b) or (c) stated above. We will show that it can only appear in case (c) with $i=0$.

In cases (a) or (b), we must have $i>0$, since the variables $f_n$ and $s_{n+1, j-1}$ do not divide $\widehat{\cM}$. Then, a factor of $\widehat{\cM}$ is a monomial of a derivative $f_n^{(i)}$ or $s_{n+1, j-1}^{(i)}$ of positive order and, by Lemma \ref{lem:prod_no_interm_partes}, it contains a variable in $\Sp^{(N+n)}$ or $\Sp^{(n+1)}$, contradicting the assumption on the variables involved in $\cM$.
It follows that $\widehat{\cM}$ is a monomial in a product in (c).

Assume that $i\geq 1$. If $h=0$, then $s_{n-1, L_{n-1}}$ divides $\mathcal{M}$ and $\widetilde{\cM}:= s_{n, L_n-1}. \frac{\cM}{s_{n-1, L_{n-1}}}$ is a monomial of $s_{n, L_n-1}^{(i)}$. Due to Lemma \ref{lem:prod_no_interm}, $\widetilde{\cM}$ contains two variables corresponding to species that react together.  By the second assumption of the lemma and the fact that $S_{n, L_n-1}$ only reacts with $S_{n-1, L_{n-1}}$ or $F_n$ (and $f_n$ does not divide $\widehat{\cM}$), one of these variables must be $s_{n-1, L_{n-1}}$; but $s_{n-1, L_{n-1}}$ does not divide $\widetilde{\cM}$, since it is square-free. If $h\geq 1$ and $\widehat{\cM} = \cM_1 \cdot \cM_2$, where $\cM_1$ is a monomial in $s_{n-1, L_{n-1}}^{(h)}$ and $\cM_2$ is a monomial in $s_{n, L_n-1}^{(i)}$, by Lemma \ref{lem:prod_no_interm}, each of the monomials $\cM_1$ and $\cM_2$ contains two variables corresponding to species that react together. One of these variables must be $s_{n, L_n-1}$, because $\cM$ does not contain two pairs of variables corresponding to species that react together. Since $S_{n, L_n-1}$ only reacts with $S_{n-1, L_{n-1}}$ or $F_n$, this is only possible in the case where $s_{n-1, L_{n-1}}$ divides $\cM$, but then $s_{n-1, L_{n-1}}$ does not divide $\frac{\cM}{s_{n-1, L_{n-1}}}$ and it does not contain two variables corresponding to species that react together.

Then, necessarily $i=0$ and $\cM$ is a monomial of $s_{n-1, L_{n-1}}^{(h)}$ for $h\le \ell -2$. This implies that $\ell\ge \ell_0+2$.

Finally, let us show that $\widehat{\cM}$ effectively appears in $s_{n, L_n}^{(\ell_0 +2)}$ and compute its coefficient. Considering formula \eqref{eq:dl_last_product} for $\ell = \ell_0+2$, by our previous arguments, we have that $\widehat{\cM} = s_{n, L_n-1} \cM$ can only arise from a product $s_{n-1, L_{n-1}}^{(h)} s_{n, L_n-1}$ when $\cM$ is a monomial of $s_{n-1, L_{n-1}}^{(h)}$ and $h\le \ell_0$.
By the minimality of $\ell_0$, the only possibility is that $h= \ell_0$; moreover, if $C_\cM$ is the coefficient of $\cM$ in $s_{n-1, L_{n-1}}^{(\ell_0)}$, the coefficient of $\widehat{\cM}$ in $s_{n, L_n}^{(\ell_0 +2)}$ is $\gamma_{u_{n, L_n}, \ell_0, 0 } C_\cM= c_{n, L_n} a_{n, L_n} C_\cM$.
\end{proof}

\begin{lemma}\label{lem:recursion_sn_int}
Let $u\, \cM$ be a monomial of $s_{n-1, L_{n-1}}^{(\ell_0)}$ which is not a monomial of any derivative of $s_{n-1, L_{n-1}}$ of lower order, where $U$ is an intermediate species and $\cM$ only involves variables corresponding to species in $\Sp^{(k)}$, for $1\le k \le n-1$, and $\Sp^{(2N+1)}$. Assume  that  $\mathcal{M}$ does not involve two variables corresponding to species that react together and $s_{n-1, L_{n-1}}$ does not divide $\cM$.

Then, $\widehat \cM := s_{n, L_n-1} u\, \cM$ is a monomial of $s_{n, L_n}^{(\ell_0+2)}$ and of no lower order derivative of $s_{n, L_n}$. Moreover, if $C_\cM$ is the coefficient of $u\, \cM$ in $s_{n-1, L_{n-1}}^{(\ell_0)}$, the coefficient of $\widehat \cM$ in $s_{n, L_n}^{(\ell_0+2)}$ is $c_{n, L_n} a_{n, L_n} C_\cM$.
In addition, if $\widetilde C_\cM$ is the coefficient of $u\, \cM$ in $s_{n-1, L_{n-1}}^{(\ell_0+1)}$, the coefficient of $\widehat \cM$ in  $s_{n, L_n}^{(\ell_0+3)}$ is $c_{n, L_n} a_{n, L_n} (\widetilde C_\cM- K_{n, L_n}  C_\cM)$.
\end{lemma}

\begin{proof}
Assume $\widehat{\cM}$ is a monomial of $s_{n, L_n}^{(\ell)}$ and consider the three cases (a), (b) and (c) listed above. We will show that it can only appear in case (c) with $i=0$.

If $\widehat{\cM}$ appears from a product of type (a), (b), or (c) with $h\ge 1$ and $i\ge 1$, there is a factor of $\widehat{\cM}$ not involving intermediate species which is a monomial of a derivative of positive order of a non-intermediate species and, by Lemma \ref{lem:prod_no_interm}, this factor involves two variables of species that react together. But $\cM$ does not contain two variables of species reacting together; in addition, the only species in $\Sp^{(k)}$, for $1\le k\le n-1$, that reacts with $S_{n, L_n-1}$ is $S_{n-1, L_{n-1}}$, and $s_{n-1, L_{n-1}}$  does not divide $\cM$.

On the other hand,  $\widehat \cM$ cannot appear from cases (a) or (b) with $h=0$ or $i=0$, since none of the variables $s_{n, L_n}$, $f_n$ or $s_{n+1, j-1}$, for $1\le j \le L_{n+1}$, divides $\widehat\cM$. Finally, the assumption that $s_{n-1,L_{n-1}}$ does not divide $\cM$ implies that the monomial cannot appear in case (c) with $h=0$.

We conclude that $\widehat{\cM}$ only appears as a monomial in $s_{n-1, L_{n-1}}^{(h)} s_{n, L_n-1}$ for $1\le h\le \ell-2$, that is, when $u\cM$ is a monomial of $s_{n-1, L_{n-1}}^{(h)}$. Then, $\ell \ge \ell_0+2$.

The computation of the coefficient of $\widehat \cM$ in $s_{n, L_n}^{(\ell_0+2)}$ follows as in the proof of Lemma \ref{lem:recursion_sn}.

Finally, let us obtain the coefficient of $\widehat \cM$ in $s_{n, L_n}^{(\ell_0+3)}$. As shown before, in formula \eqref{eq:dl_last_product} the monomial $\widehat \cM$ may only appear from terms of the form (c) with $i=0$ and $1\le h\le \ell_0+1$ such that $u\cM$ is a monomial of $s_{n-1, L_{n-1}}^{(h)}$. By the minimality of $\ell_0$, the only possible values of $h$ are $\ell_0$ and $\ell_0+1$; thus, the corresponding coefficient is
$\gamma_{u_n,L_n, \ell_0-1, 0} \widetilde C_\cM + \gamma_{u_n,L_n, \ell_0-2,0} C_\cM = c_{n, L_n} a_{n, L_n} \widetilde C_\cM  + c_{n, L_n} a_{n, L_n} (-K_{n,L_n})C_\cM  =
c_{n, L_n} a_{n, L_n} ( \widetilde C_\cM -K_{n,L_n}C_\cM)$.
\end{proof}

\begin{lemma}\label{lem:tilde_ai_Ki}
For $1\leq l\leq L_m-1$,
\[\cM_{n, s_{m,l}}=  s_{m,l }\, f_m\, s_{m,L_m} \prod\limits_{i=m+1}^n s_{i,L_i-1} \quad \hbox{ and }\quad \cM_{n, v_{m,l}}= v_{m,l}\, s_{m,L_m}\prod\limits_{i=m+1}^n s_{i,L_i-1} \]
are monomials of $s_{n,L_n}^{(2(n-m+1))}$  for every $n\ge m+1$, and they are not monomials of any derivative of $s_{n, L_n}$ of lower order. The corresponding coefficients are, respectively,
\[\tilde{a}_{m,l}\,\tilde{a}_{m,L_m}\prod_{i=m+1}^n c_{i,L_i} a_{i, L_i}\quad \hbox{ and }\quad
\tilde{K}_{m,l}\,\tilde{a}_{m,L_m} \prod_{i=m+1}^n c_{i,L_i} a_{i,L_i}.\]
\end{lemma}

\begin{proof}
For $n=m+1$, we must show that, for every $1\le l\le L_{m}-1$, $$\cM_{m+1,s_{m,l}} = s_{m,l} f_m s_{m, L_m} s_{m+1, L_{m+1}-1} \hbox{ and } \cM_{m+1,v_{m,l}} = v_{m,l} s_{m, L_m} s_{m+1, L_{m+1}-1}$$ are monomials of $s_{m+1, L_{m+1}}^{(4)}$ and of no lower order derivative of $s_{m+1, L_{m+1}}$.

It is easy to see that none of the required monomials appears in $\dot{s}_{m+1, L_{m+1}}$ or $\ddot{s}_{m+1, L_{m+1}}$, because these derivatives do not contain monomials of degree $4$ and the monomials that are multiples of intermediates have degree at most $2$ (see the proof of Lemma \ref{lem:prod_1_interm}).

Consider now the expression of $s_{m+1, L_{m+1}}^{(\ell)}$ following \eqref{eq:dl_last_product}, with $\ell \ge 3$.

The monomials $\cM_{m+1,s_{m,l}}$ and $\cM_{m+1,v_{m,l}}$ do not arise from products of type (a) or (b) with $h=0$ or $i=0$, since they are not multiples of $s_{m+1, L_{m+1}}$, $f_{m+1}$ or $s_{m+2, j-1}$. Taking into account that every monomial in a first-order derivative of a non-intermediate is either a multiple of the non-intermediate or an intermediate that reacts to it, we have that the monomials do not appear either from products of type (a) or (b) with $h=1$ or $i=1$.
As $h+i\le \ell-1$ in products of type (a) or (b), we deduce that $\cM_{m+1,s_{m,l}}$ and $\cM_{m+1,v_{m,l}}$ do not appear in these products for $\ell=3$ nor $\ell=4$.

In products of type (c), if $h+i\le 1$, there are no monomials of degree $4$, and those that are multiples of an intermediate have degree at most $2$.

We conclude that $\cM_{m+1,s_{m,l}}$ and $\cM_{m+1,v_{m,l}}$ are not monomials of $s_{m+1, L_{m+1}}^{(3)}$ and that they may only appear in $s_{m+1, L_{m+1}}^{(4)}$ from products of type (c) with $h+i =2$.

\begin{itemize}
\item $h=0$, $i=2$. By looking  at the expansion of $s_{m+1, L_{m+1}-1}^{(2)}$,  we deduce that $s_{m,l} f_m s_{m+1, L_{m+1}-1}$  and $v_{m,l} s_{m+1, L_{m+1}-1} $, for $l<L_m$, are not monomials of this derivative.
\item $h=i=1$: The monomials $\cM_{m+1,s_{m,l}}$ do not appear in this product because the only variable involved that reacts with $S_{m+1, L_{m+1}-1}$ is $S_{m, L_m}$ and the monomials $s_{m,l} f_m$ do not appear in $\dot{s}_{m, L_m}$ for $l<L_m$. The monomials $\cM_{m+1,v_{m,l}}$  do not appear since $v_{m,l}$ does not react to $s_{m, L_m}$ or $s_{m+1, L_{m+1}-1}$ for $l<L_m$.
\item $h=2$, $i=0$: As in the proof of Lemma  \ref{lem:tildeajKj}, it follows that $s_{m,l} f_m s_{m, L_m}$ and  $v_{m,l} s_{m, L_m}$ are monomials of $s_{m, L_m}^{(2)}$ with respective coefficients  $\tilde{a}_{m,l}\tilde{a}_{L_m} $ and $\tilde{K}_{m,l}\tilde{a}_{m,L_m}$.
\end{itemize}
Therefore, $\cM_{m+1,s_{m,l}}$ and $\cM_{m+1,v_{m,l}}$ effectively appear in $s_{m+1, L_{m+1}}^{(4)}$; more precisely, they arise from the product $\gamma_{u_{m+1},L_{m+1}, 2, 0} s_{m, L_m}^{(2)} s_{m+1, L_{m+1}-1}$. The corresponding coefficients can be obtained from the fact that $\gamma_{u_{m+1},L_{m+1}, 2, 0} = c_{{m+1}, L_{m+1}}.a_{m+1, L_{m+1}}$.

\smallskip

Let $n>m+1$ and assume  the monomials $\cM_{n-1, s_{m,l}}$ and  $\cM_{n-1, v_{m,l}}$ appear in $s_{n-1, L_{n-1}}^{(2(n-m))}$ and in no derivative of $s_{n-1, L_{n-1}}$ of a lower order.

Let $1\le l \le L_m-1$. Consider first $\cM_{n, s_{m,l}}$, which is a product of non-intermediates. If it appears in a derivative $s_{n, L_n}^{(\ell)}$, it arises from a product in case (a), (b) or (c) listed previously.

Since $\cM_{n, s_{m,l}}$ does not contain any variable corresponding to a species in $\Sp^{(N+n)}= \{ F_n\}$ or $\Sp^{(n+1)}= \{ S_{n+1, j}, 0\le j \le L_{n+1}\}$, by Lemma \ref{lem:prod_no_interm_partes}, it cannot appear from cases (a) or (b). Then, it is a monomial in a product $s_{n-1, L_{n-1}}^{(h)} s_{n, L_n-1}^{(i)}$ for $h+i \le \ell-2$. If $i>0$, the factor $\cM_1$ of $\cM_{n, s_{m,l}}$ which is a monomial in $s_{n, L_n-1}^{(i)}$ contains a variable in $\Sp^{(n)}$, namely $s_{n, L_n-1}$, and another variable in a set $\Sp^{(k)}$ that contains a species reacting with $S_{n,L_n-1}$. Since the only species that react with $S_{n,L_n-1}$ are $S_{n-1, L_{n-1}}$ and $F_n$, it follows that $\cM_1$ contains a variable in $\Sp^{(n-1)}$. Now, $\cM_{n, s_{m,l}}/\cM_1$ is a monomial in $s_{n-1, L_{n-1} -1}^{(h)}$; therefore, it also contains a variable in $\Sp^{(n-1)}$. But, since $n>m+1$, the only factor of $\cM_{n, s_{m,l}}$ in $\Sp^{(n-1)}$ is $s_{n-1, L_{n-1}-1}$, leading to a contradiction. We conclude that $i=0$ and $\cM_{n, s_{m,l}}$ appears as a monomial in $s_{n-1, L_{n-1}}^{(h)} s_{n, L_n-1}$, namely $\cM_{n, s_{m,l}}= \cM_{n-1, s_{m,l}} s_{n, L_n-1}$ with $\cM_{n-1, s_{m,l}}$ a monomial in $s_{n-1, L_{n-1}}^{(h)}$ for $h\le \ell-2$. Then $\ell \ge 2(n-m+1)$.

Now, consider  $\cM_{n, v_{m,l}}$ and assume it is a monomial of $s_{n, L_n}^{(\ell)}$. As, for $n>m+1$, none of the variables $s_{n, L_n}$, $s_{n+1, j-1}$, $f_n$ or $s_{n-1, L_{n-1}}$ divides $\cM_{n, v_{m,l}}$, this monomial cannot arise from cases (a) or (b) with either $i=0$ or $h=0$, nor from (c) with $h=0$. If it arises from cases (a), (b) or (c) with $h\ge 1$ and $i\ge 1$, then $\cM_{n, v_{m,l}} = \cM_1 \cM_2$ with $\cM_1$ and $\cM_2$ monomials appearing in derivatives of positive order of non-intermediate species. Assume $v_{m,l}$ divides $\cM_1$. Then, $\cM_2$ is a product of non-intermediates; by Lemma \ref{lem:prod_no_interm}, it contains the only two variables of $\cM_{n, v_{m,l}}$,  $s_{m, L_m}$ and $s_{m+1, L_{m+1}-1}$, corresponding to species that react together. On the other hand, $\cM_1 = v_{m, l} \cM$, where $\cM$ is not constant since $V_{m,l}$ does not react to $S_{n,L_n}$, $F_n$, $S_{n+1, j-1}$, $S_{n-1,L_{n-1}}$ nor $S_{n,L_n-1}$ (so, $v_{m,l}$ is not a monomial in a derivative of $s_{n,L_n}$, $f_n$, $s_{n+1, j-1}$, $s_{n-1,L_{n-1}}$ nor $s_{n,L_n-1}$). By Lemma \ref{lem:prod_1_interm}, taking into account that $V_{m,l}$ is only involved in the reactions $F_m + S_{m, l} \rightleftarrows V_{m,l} \rightarrow F_m + S_{m, l-1}$, we have that $\cM$ contains either two variables corresponding to species that react together or it contains one variable that reacts with $F_m$, $S_{m, l-1}$ or $S_{m, l}$. But none of these possibilities happen.

We conclude that $\cM_{n, v_{m,l}}$ arises from (c) with $i=0$ and it appears in $s_{n-1, L_{n-1}}^{(h)} s_{n, L_n-1}$ for $h\le \ell-2$, that is, $\cM_{n-1, v_{m,l}}$ is a monomial of $s_{n-1, L_{n-1}}^{(h)}$ for $h\le \ell -2$.  Then $\ell \ge 2(n-m+1)$.

The fact that the monomials effectively appear in $s_{n, L_n}^{2(n-m+1)}$ and the computation of their coefficients follow similarly as in the proof of Lemma \ref{lem:recursion_sn}.
\end{proof}

From the previous lemmas and the results for the case of a single layer proved in Proposition \ref{prop:2comp}, we obtain the following proposition that leads to our identifiability result for the cascade (see Table \ref{table:cascade}). The highlighted constant in each case is the one we will identify from the corresponding coefficient.

\begin{proposition} \label{propApp:monomials}
For network~\eqref{eq:cascade}, for every $n\ge m$, the following monomials $\cM$ appear in $s_{n, L_n}^{(\ell)}$ with coefficient $\pm C_\cM$ for the stated value $\ell$, and they do not appear in any derivative of $s_{n, L_n}$ of lower order:
\begin{enumerate}
\item $\cM = f_m \,s_{m,L_m} \prod\limits_{i=m+1}^n s_{i,L_i-1}$, $C_\cM = \colorbox{green}{$\tilde{a}_{m,L_m}$}\prod\limits_{i=m+1}^n c_{i,L_i} a_{i, L_i}$, $\ell = 2(n-m) +1$; \label{item:tamLm}
\item for $0\le k \le L_m-1$, $\cM = s_{m-1, L_{m-1}}^k  u_{m, L_m-k} \prod\limits_{i=m+1}^n s_{i,L_i-1}$, \\ $C_\cM = \colorbox{green}{$c_{m,L_m-k}$}\Big(\prod\limits_{j=0}^{k-1}a_{m, L_m-j} \, c_{m,L_m-j}\Big)\Big( \prod\limits_{i=m+1}^n c_{i,L_i} a_{i, L_i}\Big)$, $\ell = 2(n-m) +2k +1$; \label{item:cmj}
\item for $0\le k \le L_m-1$, $\cM = s_{m-1, L_{m-1}}^k v_{m,L_m-k} \prod\limits_{i=m+1}^n s_{i,L_i-1}$, \\ $C_\cM = \colorbox{green}{$\tilde{b}_{m,L_m-k}$}\Big(\prod\limits_{j=0}^{k-1}a_{m, L_m-j} \, c_{m,L_m-j}\Big)\Big( \prod\limits_{i=m+1}^n c_{i,L_i} a_{i, L_i}\Big)$, $\ell = 2(n-m)+2k +1$; \label{item:tbmj}
\item $\cM = s_{m-1, L_{m-1}} \prod\limits_{i=m}^n s_{i,L_i-1}$, $C_\cM = c_{m,L_m}\colorbox{green}{$a_{m,L_m}$}\prod\limits_{i=m+1}^n c_{i,L_i} a_{i, L_i}$, $\ell = 2(n-m)+2$; \label{item:amLm}
\item for $1\le j \le L_m-1$, $\cM = s_{m,j }\, f_m \,s_{m,L_m} \prod\limits_{i=m+1}^n s_{i,L_i-1}$, $C_\cM = \colorbox{green}{$\tilde{a}_{m,j}$} \tilde{a}_{m,L_m}\prod\limits_{i=m+1}^n c_{i,L_i} a_{i, L_i}$, $\ell = 2(n-m)+2$; \label{item:tamj}
\item for $1\le j \le L_m-1$, $ \cM = v_{m,j} \, s_{m,L_m}\prod\limits_{i=m+1}^n s_{i,L_i-1}$,  $C_\cM = \colorbox{green}{$\tilde{K}_{m,j}$}\tilde{a}_{m,L_m} \prod\limits_{i=m+1}^n c_{i,L_i} a_{i,L_i}$, $\ell = 2(n-m)+2$;
      \label{item:tKmj}
\item for $1\le j \le L_{m}-1$, $\cM = s_{m, j-1} s_{m-1, L_{m-1}}\prod\limits_{i=m}^n s_{i,L_i-1}$,  $C_\cM = \colorbox{green}{$a_{m, j}$} \prod\limits_{i=m}^n c_{i,L_i} a_{i, L_i}$, $\ell = 2(n-m)+3$; \label{item:amj}
\item for $1\le j \le L_{m}-1$, $\cM =  u_{m, j} \prod\limits_{i=m}^n s_{i,L_i-1}$, $C_\cM = \colorbox{green}{$K_{m, j}$} \prod\limits_{i=m}^n c_{i,L_i} a_{i, L_i}$, $\ell = 2(n-m)+3$. \label{item:Kmj}
\end{enumerate}
Furthermore, the monomials
$ u_{m, L_m} \prod\limits_{i=m+1}^n s_{i,L_i-1}$ and  $v_{m, L_m} \prod\limits_{i=m+1}^n s_{i,L_i-1}$ (c.f. items \ref{item:cmj} and \ref{item:tbmj} with $k=0$)
appear in $s_{n, L_n}^{(2(n-m)+2)}$  with coefficients
$ -\Big(\colorbox{green}{$K_{m, L_m}$} + \sum\limits_{i=m+1}^n K_{i, L_i}\Big)c_{m,L_m}\Big( \prod\limits_{i=m+1}^n c_{i,L_i} a_{i, L_i}\Big)$ and $-\Big(\colorbox{green}{$\tilde{K}_{m, L_m}$}+ \sum\limits_{i=m+1}^n K_{i, L_i}\Big)\tilde{b}_{m,L_m}\Big( \prod\limits_{i=m+1}^n c_{i,L_i} a_{i, L_i}\Big)$ respectively.
\end{proposition}

\begin{proof}
Fix $m$ with $1\le m\le N$. We prove the proposition inductively for $n\ge m$.

The case $n=m$ is considered in Subsection \ref{subsec:2comp}.

Let $n\ge m+1$. Items \ref{item:tamj} and \ref{item:tKmj} are proved in Lemma \ref{lem:tilde_ai_Ki}. 
For the remaining monomials, assuming the statement holds for $n-1$, we deduce that it is also true for $n$ by applying Lemma \ref{lem:recursion_sn} (for items \ref{item:tamLm}, \ref{item:amLm} and \ref{item:amj}) and Lemma \ref{lem:recursion_sn_int} (for items \ref{item:cmj}, \ref{item:tbmj}, \ref{item:Kmj}, and the last statement of the proposition).

We present a complete proof in the first two cases. The induction step for the monomials of the remaining items follows similarly.

\begin{enumerate}
\item Consider $\cM_1 =f_m \,s_{m,L_m}\prod\limits_{i=m+1}^{n-1} s_{i,L_i-1}$. By the inductive assumption, this monomial appears in $s_{n-1,L_{n-1}}^{(2(n-1-m)+1)}$ with coefficient $C_{\cM_1} = \tilde{a}_{m,L_m}\prod\limits_{i=m+1}^{n-1} c_{i,L_i} a_{i, L_i}$, and in no derivative of $s_{n-1, L_{n-1}}$ of a lower order.

    Let us show that $\cM_1$ satisfies the assumptions of Lemma \ref{lem:recursion_sn}. First, note that $\cM_1$ is square-free and only involves variables corresponding to species in $\Sp^{(k)}$, for $m\le k \le n-1$, and $\Sp^{(N+m)}$. In addition, since two species $S_{i,L_{i}-1}$, $S_{j,L_{j}-1}$, for $m\le i,j\le n-1$,  do not react together and $F_m$ does not react with $S_{i, L_i-1}$ for $m+1\le i \le n-1$, then $\cM_1$ does not contain two disjoint pairs of variables corresponding to species that react together. Finally, we have that $s_{n-1, L_{n-1}}$ divides $\cM_1$ only when $n= m+1$, and in this case, $\cM_{1} = f_{n-1} s_{n-1, L_{n-1}}$, which clearly satisfies the assumptions of the lemma.

    Therefore, by Lemma \ref{lem:recursion_sn}, we conclude that $s_{n, L_n-1} \cM_1 = f_m \,s_{m,L_m}\prod\limits_{i=m+1}^{n} s_{i,L_i-1}$ is a monomial of $s_{n,L_n}^{(2(n-1-m)+1+2)}= s_{n,L_n}^{(2(n-m)+1)}$ and of no lower order derivative of $s_{n, L_n}$, and its corresponding coefficient is $c_{n, L_n}a_{n,L_n}C_{\cM_1} = \tilde{a}_{m,L_m}\prod\limits_{i=m+1}^{n} c_{i,L_i} a_{i, L_i}$.

\item For a fixed $k$, with $0\le k \le L_m-1$, the monomial $\cM$ can be written as $\cM= s_{n, L_{n}-1}u\cM_2$, where $u:=u_{m, L_m-k}$ is a variable corresponding to an intermediate species and $\cM_2:=  s_{m-1, L_{m-1}}^k  \prod\limits_{i=m+1}^{n-1} s_{i,L_i-1}$.

    By the induction assumption, we have that $u\cM_2$ is a monomial of $s_{n-1, L_{n-1}}^{(2(n-1-m)+2k+1)}$, with coefficient $C_{\cM_2} = c_{m,L_m-k}\Big(\prod\limits_{j=0}^{k-1}a_{m, L_m-j} \, c_{m,L_m-j}\Big)\Big( \prod\limits_{i=m+1}^{n-1} c_{i,L_i} a_{i, L_i}\Big)$, and it does not appear in any lower order derivative of $s_{n-1, L_{n-1}}$.

    Let us show that $\cM_2$ satisfies the assumptions of Lemma \ref{lem:recursion_sn_int}. It is clear that $\cM_2$ only involves variables in $\Sp^{(i)}$ for $i \le n-1$ and $\Sp^{(2N+1)}$, and that $s_{n-1, L_{n-1}}$ does not divide $\cM_2$, since $n\ne m$. Also, since two species $S_{i,L_{i}-1}$ and $S_{j,L_{j}-1}$, for $m\le i,j\le n-1$,  do not react together and $S_{m-1, L{m-1}}$ does not react with $S_{i, L_{i-1}}$ for $i\ge m+1$,  it follows that $\cM_2$ does not involve two variables corresponding to species that react together.

    Then, by Lemma \ref{lem:recursion_sn_int}, we conclude that $s_{n, L_n-1} u\cM_2 =  s_{m-1, L_{m-1}}^k  u_{m, L_m-k} \prod\limits_{i=m+1}^n s_{i,L_i-1}$ is a monomial of $s_{n, L_n}^{(2(n-m)+2k+1)}$ and of no lower order derivative of $s_{n, L_n}$, and its  coefficient is $c_{n,L_n}a_{n,L_n} C_{\cM_2} =  c_{m,L_m-k}\Big(\prod\limits_{j=0}^{k-1}a_{m, L_m-j} \, c_{m,L_m-j}\Big)\Big( \prod\limits_{i=m+1}^{n} c_{i,L_i} a_{i, L_i}\Big)$.
\end{enumerate}
\end{proof}

\end{document}